\documentclass[11pt]{article}
\usepackage[margin=1in]{geometry}
\usepackage[T1]{fontenc}
\usepackage{lmodern}
\usepackage{amsmath,amssymb,amsthm,mathtools,booktabs,array,microtype,enumitem}
\usepackage{xcolor}
\usepackage{tikz}
\usepackage{float,needspace}
\usetikzlibrary{arrows.meta}
\usepackage[colorlinks=true,linkcolor=blue!45!black,citecolor=blue!45!black,urlcolor=blue!45!black]{hyperref}
\usepackage{bookmark}
\hypersetup{pdftitle={Fine-Grained Hardness of Approximating Dynamic Time Warping},pdfauthor={Jihan Wang}}
\newtheorem{theorem}{Theorem}[section]
\newtheorem{lemma}[theorem]{Lemma}
\newtheorem{proposition}[theorem]{Proposition}
\newtheorem{corollary}[theorem]{Corollary}

\theoremstyle{definition}
\newtheorem{definition}[theorem]{Definition}
\newtheorem{hypothesis}[theorem]{Hypothesis}
\newtheorem{remark}[theorem]{Remark}
\newcommand{\DTW}{\operatorname{DTW}}
\newcommand{\shrink}{\operatorname{shrink}}
\newcommand{\bits}{\{0,1\}}
\newcommand{\ind}{\mathbf{1}}
\newcommand{\str}[1]{\mathtt{#1}}
\DeclareMathOperator{\polylog}{polylog}
\setlist{topsep=4pt,itemsep=3pt}
\title{Fine-Grained Hardness of Approximating Dynamic Time Warping}
\author{Jihan Wang\\[0.3em]
    \small The University of Illinois Urbana-Champaign\\
    \small Department of Computer Science\\
    \small Urbana, IL, USA\\
    \small \href{mailto:jihanw2@illinois.edu}{\nolinkurl{jihanw2@illinois.edu}}}
\date{}
\begin{document}
\maketitle
\begin{abstract}
We prove conditional lower bounds for approximating dynamic time
warping (DTW) over the uniform metric on three symbols.
Let $N,M$ be the two string lengths and
$n,m$ their respective numbers of runs.
Under the Gap Block Disjointness (GBD) hypothesis, for every fixed
$\kappa\in(0,1/6)$ and $\delta\in(0,1)$, no deterministic algorithm
$N^{\kappa\delta}$-approximates DTW on two explicit strings of
length $N$ in $O(N^{2-\delta})$ time. This matches the known
tradeoff between approximation and running-time exponents for
deterministic algorithms within any fixed factor greater than six in
the approximation exponent.
For run-length encoded strings, the Orthogonal Vectors
Hypothesis (OVH) rules out any constant-factor approximation in
$\widetilde O((nm)^{1-\delta})$ time for every constant $\delta>0$.
More generally, for every fixed $\eta\in(0,1)$, it rules out
$(N+M)^{1-\eta}$-approximation in the same running time.
These bounds extend to every fixed metric with at least three points, including
absolute distance on $\{0,1,2\}$.

Both reductions use long runs to force equal-symbol matches in
low-cost alignments. For explicit strings, we combine Boolean gadgets
with deterministic gap amplification to obtain a polynomial
approximation gap. For run-length encoded strings, an
encoding of regular-expression membership gives accepting instances
an alignment whose cost does not increase as selected runs grow.
Rejecting instances have distance at least the chosen run length,
while the number of runs stays fixed.
\end{abstract}

\section{Introduction}

Dynamic time warping (DTW) measures string similarity by minimizing
the cost of a monotone alignment that may repeat symbols. Let $X,Y$
be nonempty strings of lengths $N,M$, respectively. A \emph{warping path}
$\pi$ is a sequence of cells in the $N\times M$ grid. It starts at
$(1,1)$ and ends at $(N,M)$, using steps $(1,0),(0,1),(1,1)$.
Under the uniform metric, its cost
is the number of visited cells with unequal symbols. The DTW distance is
\[
 \DTW(X,Y)=\min_\pi\sum_{(i,j)\in\pi}\ind[X[i]\ne Y[j]].
\]

The standard dynamic program computes DTW in $O(NM)$ time.
Under the Strong Exponential Time
Hypothesis (SETH)~\cite{ImpagliazzoPaturi2001}, no deterministic
strongly subquadratic algorithm
computes DTW exactly~\cite{AbboudBackursWilliams2015,BringmannKunnemann2015}.
Kuszmaul~\cite{Kuszmaul2019} proved this bound even for equal-length
strings over a three-symbol uniform metric.

For two strings of length $N$ over a uniform metric,
Kuszmaul~\cite{Kuszmaul2019} gave a deterministic
$O(N^\varepsilon)$-approximation in $\widetilde O(N^{2-\varepsilon})$
time\footnote{Throughout, $\widetilde O$ hides factors polylogarithmic
in the total expanded length $N+M$.}
for every fixed $\varepsilon\in(0,1)$. How much approximation is
necessary to obtain a polynomial saving over quadratic time in $N$?
Under the Gap Block Disjointness hypothesis, we show that such
a saving requires a polynomial approximation factor for deterministic
algorithms, even over three symbols.

A second approximation question arises for run-length encoded (RLE) strings.
This representation stores each maximal run of equal symbols by its
symbol and length. Let $n,m$ be the respective numbers of runs in
$X,Y$. Exact DTW on run-length encoded strings can be computed in
$\widetilde O(nm)$ time~\cite{BonehEtAl2024}. Under SETH, this is
optimal up to subpolynomial factors for deterministic
algorithms~\cite{AbboudBackursWilliams2015,BringmannKunnemann2015}.
Can approximation give a fixed polynomial improvement over this
running time?
Under the Orthogonal Vectors Hypothesis (OVH),
we rule out such an improvement for every constant approximation factor.

\subsection{Our results}

\paragraph{Explicit strings.}
Our first result assumes the Gap Block Disjointness (GBD) hypothesis
of Abboud and Backurs~\cite{AbboudBackurs2017}, introduced to study
deterministic approximation of LCS.

\begin{theorem}\label{thm:explicit}
Assume the Gap Block Disjointness hypothesis
(Hypothesis~\ref{hyp:gbd}).
For every fixed $\kappa\in(0,1/6)$ and $\delta\in(0,1)$,
no deterministic algorithm
$N^{\kappa\delta}$-approximates DTW in $O(N^{2-\delta})$ time.
This holds for two strings of length $N$ over the uniform metric on
$\Sigma=\{\str{a},\str{b},\str{c}\}$.
\end{theorem}

In particular, GBD excludes deterministic strongly subquadratic
algorithms with any $N^{o(1)}$ approximation factor. Together with
Kuszmaul's algorithm, this gives a tradeoff between the approximation
and running-time exponents that is tight up to any fixed factor greater
than six in the approximation exponent.

\paragraph{Run-length encoded strings.}
Our second result assumes the Orthogonal Vectors
Hypothesis~\cite{Williams2005}.

\begin{theorem}\label{thm:rle}
Assume the Orthogonal Vectors Hypothesis (Hypothesis~\ref{hyp:rovh}).
For every $\delta>0$, the following lower bounds hold for
run-length encoded DTW over the uniform metric on $\Sigma$.
\begin{enumerate}[label=(\roman*),ref=\roman*]
\item\label{thm:rle-constant}
For every constant $C\geq1$, no algorithm $C$-approximates DTW in
$\widetilde O((nm)^{1-\delta})$ time.
This holds even when every run has length one or $H$, where
$H=\Theta_C(n+m)$ and
$N+M=\Theta_C((n+m)^2)$.
\item\label{thm:rle-large}
For every fixed $\eta\in(0,1)$, no algorithm
$(N+M)^{1-\eta}$-approximates DTW in
$\widetilde O((nm)^{1-\delta})$ time.
This holds even when every run has length one or $H$ and
$N+M\leq(n+m)^{O(1/\eta)}$.
\end{enumerate}
\end{theorem}

An elementary algorithm compares the two sequences of run labels.
It returns zero if they agree and $N+M$ otherwise. This gives an
$(N+M)$-approximation in $O(n+m)$ time.
Part~(\ref{thm:rle-large}) excludes every fixed polynomial
improvement in this factor in $\widetilde O((nm)^{1-\delta})$
time for any fixed $\delta>0$.
It also excludes every fixed polynomial approximation factor in the
total number of runs $n+m$.

The alphabet size is sharp for both results: binary DTW has a
linear-time algorithm on explicit strings and a near-linear-time
algorithm in the number of runs on RLE inputs~\cite{Kuszmaul2021Binary}.
Both lower bounds extend to every fixed metric with at least three
points, including absolute distance on $\{0,1,2\}$
(Remark~\ref{rem:fixed-metrics}).

The explicit-string construction also gives circuit consequences
without assuming GBD. Here $v$ is the input length, and
$\mathsf{E}^{\mathsf{NP}}$ denotes deterministic time $2^{O(v)}$
with an NP oracle.
For fixed $\kappa\in(0,1/6)$ and
$\delta\in(0,1)$, suppose a deterministic algorithm
$N^{\kappa\delta}$-approximates DTW on two strings of length $N$ in
$O(N^{2-\delta})$ time. For every polynomial size bound and every
polynomial-time computable sublinear depth bound, some language in
$\mathsf{E}^{\mathsf{NP}}$ then has no nonuniform fan-in-two circuits
meeting these bounds. By the connection of Abboud and
Backurs~\cite{AbboudBackurs2017}, the assumed algorithm also rules out
nonuniform linear-size Valiant series-parallel circuits for
$\mathsf{E}^{\mathsf{NP}}$.

\subsection{Related work}

In the one-way communication model, Braverman, Charikar, Kuszmaul, and
Yang~\cite{BravermanEtAl2020} proved an $\Omega(N/\alpha)$-bit lower
bound for randomized protocols that $\alpha$-approximate DTW.
The bound holds for strings of length at most $N$, for
$1\le\alpha\le N$, over any metric with at least three points.
Their protocols match this bound up to polylogarithmic factors for
metrics of polynomial size and aspect ratio.
Their three-symbol uniform-metric hard instances use long runs to
force ordered equal-symbol matches and separate distance at most one
from distance above $\alpha$.
Transferring such gaps to running-time lower bounds also requires
a fine-grained source reduction with instances compact enough to
preserve the running-time saving.

Froese, Jain, Rymar, and Weller~\cite{FroeseEtAl2023} gave an
exact algorithm for RLE time series in $O(n^2m+nm^2)$ time.
For RLE DTW over the uniform metric, Xi and
Kuszmaul~\cite{XiKuszmaul2022} gave a $(1+\varepsilon)$-approximation
in $\widetilde O(nm/\varepsilon^2)$ time.
Boneh et al.~\cite{BonehEtAl2024} subsequently gave an exact
$\widetilde O(nm)$-time algorithm for arbitrary costs.
For Euclidean DTW, Agarwal, Fox, Pan, and
Ying~\cite{AgarwalEtAl2016} gave near-linear-time
$(1+\varepsilon)$-approximations for point sequences drawn from
packed or bounded curves, for fixed dimension and curve parameters.
Boneh, Golan, and Kociumaka~\cite{BonehGolanKociumaka2026} recently
proved an $(N+M)^{1-o(1)}k$ lower bound for exact
DTW on explicit strings in the low-distance regime with threshold
$k$, under OVH. This matches Kuszmaul's
$O((N+M)k)$ upper bound~\cite{Kuszmaul2019} up to subpolynomial
factors. Their reduction separates consecutive distance
values at a growing threshold. Our RLE reduction separates a \textsc{Yes}
bound of $n+m-1$ from a freely chosen threshold $H$.

Under GBD, Abboud and Backurs~\cite{AbboudBackurs2017} showed that
there is no $\delta>0$ such that, for every $\varepsilon>0$, LCS has
a deterministic $(1+\varepsilon)$-approximation in
$O(N^{2-\delta})$ time.
Their construction uses a growing alphabet to select a corresponding
row pair. Our Boolean gadgets implement this selection over three
symbols. Such alphabet growth is necessary for a polynomial
approximation gap for LCS, since it admits a linear-time approximation
whose factor is the alphabet size.
Abboud and Rubinstein~\cite{AbboudRubinstein2018} subsequently showed
that deterministic strongly subquadratic constant-factor approximation
for LCS over a subpolynomial-size alphabet would imply new circuit
lower bounds. Chen et al.~\cite{ChenEtAl2019} strengthened these
consequences to bounded-fan-in circuits of depth $v^{o(1)}$ and
nonuniform $\mathsf{NC}^1$.

Approximation behaves differently for other similarity measures.
Buchin, Ophelders, and Speckmann~\cite{BuchinOpheldersSpeckmann2019}
strengthened earlier SETH-based lower
bounds~\cite{Bringmann2014Frechet,BringmannMulzer2016} to rule out
deterministic strongly subquadratic $(3-\varepsilon)$-approximation
for both continuous and discrete Fr\'echet distance.
This holds for any fixed $\varepsilon>0$, even in one dimension.
A recent deterministic strongly subquadratic
$(3+\varepsilon)$-approximation for Fr\'echet distance
in fixed-dimensional Euclidean space nearly matches this
bound~\cite{VanDerHorstOphelders2026}.
For unit-cost edit distance, no fine-grained lower bound is known
that excludes strongly subquadratic $(1+\varepsilon)$-approximation
for any fixed $\varepsilon>0$.
Randomized algorithms approximate unit-cost edit distance in
subquadratic time~\cite{ChakrabortyEtAl2020,GoldenbergRubinsteinSaha2020,MaoRubinstein2026}
and achieve constant factors in near-linear
time~\cite{AndoniNosatzki2026}.
Metric weighted edit distance also admits randomized constant-factor
approximations in strongly subquadratic
time~\cite{MaderTavasoliWang2026,DasKipouridisKociumaka2026,Gorbachev2026}.

\section{Technical overview}\label{sec:overview}

Both reductions separate the number of runs from their lengths.
Long runs constrain low-cost alignments, while equal-symbol
comparisons let accepting alignments remain cheap as these runs grow.
The explicit-string reduction must also keep the expanded strings
short enough to preserve a running-time saving. The RLE reduction
can choose longer runs without increasing the run count. Its run
labels do not depend on the long-run length, so it fixes them first
and chooses this length afterwards.

Our tagged Boolean gadgets encode the source predicate. A distance
cap controls the cost of bad pairs, and disjoint strips make their
lower bounds add across the list construction.
For RLE strings, our factor encoding uses short words with designated
heavy positions. These positions force the allowed emissions without
increasing the accepting cost.

\Needspace{12\baselineskip}
\subsection{Long runs and low-cost paths}\label{sec:overview-runs}

A cell $(i,j)$ is an \emph{equality cell} if $X[i]=Y[j]$.
The \emph{run collapse}
$\shrink(X)$ replaces each run by one copy of its symbol.
Its length is the number $n$ of runs in $X$; similarly,
$|\shrink(Y)|=m$. Two strings have distance zero exactly when
their run collapses agree.

Fix an integer $H\geq2$. Our gadgets use runs of length one,
called \emph{light}, and runs of length $H$, called \emph{heavy}.
Suppose a path segment traverses every position of a heavy run.
If the segment has cost below $H$, at least one of these positions
must occur in an equality cell. Otherwise its $H$ positions give
$H$ distinct mismatching cells. Along the same warping path,
equality cells from successive heavy runs occur in string order.
When the runs have different
labels, their equality cells must also lie in distinct runs on
the opposite axis.

Equality cells have a second use. Start with a path on
run-collapsed strings, and select positions that have equality
cells. Replace each selected position by $H$ copies of its symbol.
At its equality cell, traverse the added copies while holding the
opposite coordinate fixed. All added cells have cost zero.
Repeating this on both axes gives a path of the same cost in the
expanded strings.

The two properties serve different directions of the reductions.
The first decodes a low-cost path and proves soundness. The second
constructs a low-cost path and proves completeness.

\subsection{Explicit strings}\label{sec:overview-explicit}

A GBD instance consists of two lists of $p$ Boolean matrices with $K$ rows and $L$
columns. A pair is \emph{good} if some corresponding pair of rows
is disjoint. Its predicate is
\[
 F(A,B)=\bigvee_{k=1}^K\bigwedge_{j=1}^L
                 (\neg A_{k,j}\vee\neg B_{k,j}).
\]
For $p\geq3$, the Gap Block Disjointness problem asks to distinguish
the following cases:
\begin{enumerate}[label=(\roman*)]
\item \textsc{No}: no cross-list pair is good.
\item \textsc{Yes}: at least a $1-(\log p)^{-10}$ fraction of
      cross-list pairs is good.
\end{enumerate}
We use the hypothesis of Abboud and
Backurs~\cite{AbboudBackurs2017}.\footnote{We follow their verbal
definition of a good pair. The displayed predicate and \textsc{No} case in
their definition on page 11:4 contain inconsistent signs.}

\begin{hypothesis}\label{hyp:gbd}
For every pair of constants $\varepsilon,\alpha>0$, there is a
constant $c_0>0$ such that no deterministic algorithm solves
Gap Block Disjointness on two lists of $p$ matrices with dimensions
\[
 K=\lceil p^\alpha\rceil,\qquad L=\lceil c_0\log p\rceil
\]
in $O(p^{2-\varepsilon})$ time.
\end{hypothesis}

GBD is a deterministic hypothesis: a uniformly sampled pair
distinguishes the two cases with bounded error in $O(KL)$ time.
We encode each matrix
separately as a ternary string. Good pairs must have small distance,
and bad pairs must have distance at least the freely chosen $H$.

\Needspace{14\baselineskip}
\begin{definition}\label{def:encoders}
Fix an integer $r\ge3$.
A predicate $F(u,v)$ has \emph{separate encodings with run bound $r$}
if, for every integer $H\ge2$, there are
deterministically computable maps $X_F(u)$ and $Y_F(v)$ over $\Sigma$ satisfying:
\begin{enumerate}[label=(\roman*)]
\item Each output has at most $r$ runs, all of length one or $H$.
      It begins with $\str{c}^H$ and ends with $\str{b}^H$.
\item The left and right run collapses are fixed and distinct, independent of
      the inputs and of $H$.
\item True pairs have distance at most $r$, and false pairs have
      distance at least $H$.
\item Each encoder computes its output from its own input in $O(Hr)$ time.
\end{enumerate}
\end{definition}

The fixed collapses supply common dummies. Distinct side collapses
also allow a fixed capping pair with distance at least $H$.
Equal collapses would give distance zero.

\paragraph{Boolean gates.}
We build the matrix encoding from literal encodings, AND gates,
and OR gates. For a literal known to the left encoder, the right
string is $\str{c}^H\str{b}^H$. The left string is
$\str{c}^H\str{a}\str{b}^H$ if the literal is true, and
$\str{c}^H\str{a}^H\str{b}^H$ if it is false. A true literal has
distance at most one. In a false literal, all $H$ copies of
$\str{a}$ mismatch. Exchange the sides for a literal known to the
right encoder.

For a gate with $k$ children, let $X_j,Y_j$ be the strings in
its $j$-th child comparison, $j\in[k]$.
For soundness, we must recover complete comparisons between
corresponding children from a low-cost path. Such a path may enter
a child partway through or compare it with several children on
the other side. We place heavy tags before and after each child.
A leading tag alternates $\str{a},\str{b}$, and a trailing tag
alternates $\str{a},\str{c}$. For children with at most $r$ runs,
the tags have $r+2$ heavy $\str{b}$-runs and $\str{c}$-runs,
respectively. Write $\widehat X_j,\widehat Y_j$ for the tagged
children. Each occupies one active slot. The gate begins with an
outer $\str{c}^H$ and ends with an outer $\str{b}^H$.
A light dummy slot contains a
run-collapsed copy of an opposite tagged child. The tags in this
copy have light runs. We may insert these slots between
active slots, but never two light slots consecutively.

For any path of cost below $H$, the tags give ordered, disjoint,
nonempty intervals of left slots, one for each active right child.
A singleton interval at an active left slot yields a complete
child comparison of cost below $H$. The AND layout forces a
corresponding singleton for every child, and the OR layout forces
one for at least one child.

Consider the leading tag of an active right child.
A path below cost $H$ must give
its alternating heavy runs equality cells, so the $\str{b}$-witnesses
lie in distinct left $\str{b}$-runs. One child and the outer terminal
run can supply at most $r+1$ such runs. Reaching another child
from an active slot crosses a heavy trailing $\str{c}$-run,
which would cost at least $H$. From a light slot, the path must
cross the next heavy left leading tag to reach another child
while remaining in the right leading tag. Each heavy
$\str{b}$-run it crosses must have an equality cell with the
right leading tag. Thus the right leading tag
must have an equality cell in a left leading tag, which may be
heavy or light.
The trailing tags enforce the corresponding statement for
$\str{c}$.

Use these leading and trailing witnesses as anchors for each
right child $\widehat Y_j$. The anchors define a nonempty interval
$I_j$ of left slots, including both anchor slots.
Within each left slot, the leading tag precedes the trailing tag.
On the right, a child's leading anchor precedes its trailing anchor,
which precedes the next child's leading anchor. Monotonicity makes
the intervals ordered. Two successive intervals cannot share a left
slot, since that would require visiting its trailing tag before its
leading tag. If an interval consists of a
single active left slot, the subpath between its anchors traverses both
children completely. Leading tags contain no $\str{c}$, and
trailing tags contain no $\str{b}$. A path below $H$ must
have equality cells joining the two initial
$\str{c}^H$ runs and the two terminal $\str{b}^H$ runs.
Otherwise the endpoint witnesses would cross. These cells
bound a child comparison that we can complete at zero
additional cost.

For AND, there are $k$ right children and $k$ left slots.
The $k$ ordered, disjoint, nonempty intervals must be the $k$
corresponding singletons. Thus every corresponding child pair has
distance below $H$, so every child is true.

An OR requires a common child collapse on each side.
Set $U_X=\shrink(\widehat X_j)$ and
$U_Y=\shrink(\widehat Y_j)$, independently of $j$.
The left dummies are $U_Y$ and the right dummies are $U_X$,
so they permit zero-cost comparisons with opposite tagged children.
There are $2k-1$ left slots,
with active child $j$ at slot $2j-1$ and light dummies between
active slots. On the right, a light dummy appears before the first
active child, between consecutive active children, and after the
last, giving $2k+1$ slots.
Some interval must be the corresponding singleton
$\{2j-1\}$. If the first interval is not $\{1\}$, its right endpoint is
at least two. Each later interval begins after its predecessor.
Avoiding $\{2j-1\}$ makes the $j$-th right endpoint at least
$2j$, by induction. This is impossible for $j=k$.
The resulting singleton gives $\DTW(X_j,Y_j)<H$, so child $j$ is true.
Figure~\ref{fig:or-slots}
illustrates this slot count.

\begin{figure}[htbp]
\centering
\begin{tikzpicture}[x=1.5cm,y=1cm,font=\small,
 slot/.style={draw,minimum width=1.35cm,minimum height=.6cm},
 interval/.style={thick,blue!55!black}]
\node[slot] at (0,0) {$\widehat X_1$};
\node[slot,fill=black!6] at (1.5,0) {$U_Y$};
\node[slot,fill=blue!10] at (3,0) {$\widehat X_2$};
\node[slot,fill=black!6] at (4.5,0) {$U_Y$};
\node[slot] at (6,0) {$\widehat X_3$};
\foreach \i/\x in {1/0,2/1.5,3/3,4/4.5,5/6}
  \node[font=\footnotesize] at (\x,-.55) {slot $\i$};
\draw[interval] (-.45,.65)--(-.45,.85)--(1.95,.85)--(1.95,.65);
\draw[interval] (2.55,.65)--(2.55,.85)--(3.45,.85)--(3.45,.65);
\draw[interval] (4.05,.65)--(4.05,.85)--(6.45,.85)--(6.45,.65);
\node at (.75,1.15) {$I_1=[1,2]$};
\node at (3,1.15) {$I_2=\{3\}$};
\node at (5.25,1.15) {$I_3=[4,5]$};
\end{tikzpicture}
\caption{Every family of $k$ ordered disjoint nonempty intervals in
$2k-1$ slots contains a corresponding singleton $I_j=\{2j-1\}$.
Here $I_2$ isolates the pair $(X_2,Y_2)$.}
\label{fig:or-slots}
\end{figure}
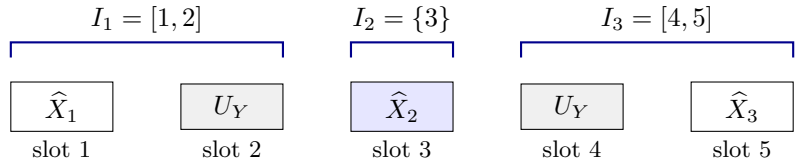

For completeness, select a true child $i$. Use a path of cost at most
$r$ on $X_i,Y_i$ and match their tags at zero cost.
Each unselected tagged child matches
an opposite dummy with the same run collapse, also at zero cost.
Only the two right dummies adjacent to $\widehat Y_i$ remain.
Traverse them while holding the first and last positions of
$\widehat X_i$, respectively. Their total length is $O(r)$, so a
true OR pair has cost $O(r)$.

We standardize the literal collapses using a two-input AND.
Its first slot is reserved for a left literal and its second for
a right literal. Put the given literal in its reserved slot and a
true literal in the other. This preserves its value and gives all
literal encodings the same collapse on each side. Applying the
gates to the three levels of $F$ gives run bound $r=O(KL)$.

\paragraph{Combining lists.}
A fixed pairing of the lists, such as equal indices, may consist
entirely of bad pairs even when most cross-list pairs are good.
If at most a $\theta$ fraction of pairs are bad, averaging over
cyclic shifts gives a shift with at most $\theta p$ bad comparisons.
Following Abboud and Backurs~\cite{AbboudBackurs2017}, we realize
every shift in one fixed pair of strings: repeat the right list
twice and place $p$ light dummy blocks on each side of the left
list. Figure~\ref{fig:explicit-overview} shows the arrangement,
with $A_i,B_j$ denoting the local encoding blocks.

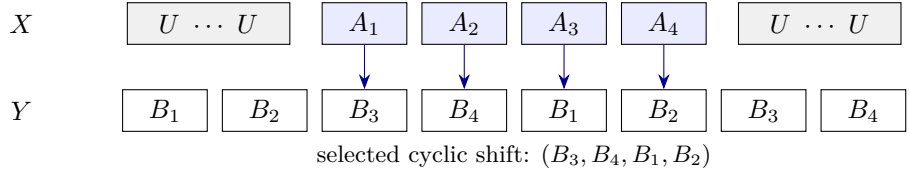
\begin{figure}[htbp]
\centering
\begin{tikzpicture}[x=1.65cm,y=1cm,font=\small,
 block/.style={draw,minimum width=1.12cm,minimum height=.55cm,inner sep=2pt},
 dummy/.style={block,fill=black!6},
 selected/.style={block,fill=blue!8},
 match/.style={-{Stealth[length=2mm]},blue!55!black}]
\node at (-.35,1) {$X$};
\node[dummy,minimum width=2.15cm] at (1.15,1) {$U\ \cdots\ U$};
\foreach \i/\x in {1/2.4,2/3.2,3/4,4/4.8}
  \node[selected] (a\i) at (\x,1) {$A_{\i}$};
\node[dummy,minimum width=2.15cm] at (6.05,1) {$U\ \cdots\ U$};
\node at (-.35,-.15) {$Y$};
\foreach \i/\j/\x in {1/1/.8,2/2/1.6,3/3/2.4,4/4/3.2,5/1/4,6/2/4.8,7/3/5.6,8/4/6.4}
  \node[block] (b\i) at (\x,-.15) {$B_{\j}$};
\foreach \i/\j in {1/3,2/4,3/5,4/6}
  \draw[match] (a\i.south)--(b\j.north);
\node[font=\footnotesize] at (3.6,-.75)
  {selected cyclic shift: $(B_3,B_4,B_1,B_2)$};
\end{tikzpicture}
\caption{The explicit construction with $p=4$, omitting all guards.
The left string contains four dummy blocks on each side of the active
list; the right string repeats its list twice. The highlighted
comparisons realize one cyclic shift. Each dummy $U$ is the common
run collapse of the right matrix strings.}
\label{fig:explicit-overview}
\end{figure}

Two additional features make this arrangement work for DTW.
First, without a cap, each bad local distance has an $O(Hr)$
upper bound. Summing over $\theta p$ bad comparisons gives an
$O(\theta pHr)$ upper bound. This does not ensure a gap from the
\textsc{No} lower bound $pH$ when $\theta r$ is large.
We cap every local distance at $H+O(r)$ by applying OR to the local
encodings and a fixed pair. The fixed pair has the same side
collapses and distance between $H$ and $H+O(r)$.
The OR preserves the $O(r)$ distance on good pairs and the
$H$ lower bound on bad pairs.

Second, the soundness argument cannot assume that a cheap path
follows a cyclic shift.
We place heavy guards around the active left blocks and the
right blocks. For each active left block, take a strip containing
that block, its immediately preceding leading tag, and its
immediately following trailing tag.
Fix an arbitrary warping path. The cost of a strip is the total cost
of the path cells in that strip. A strip costing below $H$
forces a complete comparison with one right block, hence a good
pair. If every pair is bad, each of the $p$ strips costs at least
$H$. Their left-coordinate ranges are disjoint, so these costs
add even if two strips encounter the same right block.

The cyclic-shift alignment and the disjoint-strip argument give
\[
 \textsc{Yes}:\ O(pr+\theta pH),\qquad
 \textsc{No}:\ \Omega(pH).
\]
For an integer approximation gap $q\geq1$, choose $H$ of order $qr$
and require $\theta$ to be a sufficiently small constant times
$1/q$. The \textsc{No} bound then exceeds $q$ times the \textsc{Yes} bound.

\paragraph{Amplifying the density.}
The GBD promise gives $\theta\le(\log p)^{-10}$, which guarantees
only a polylogarithmic $q$. We amplify the density to obtain a
polynomial $q$.
Pad each list to the next power-of-two size by repeating one of
its entries. Write the entries as $u_i,v_j$, with bit-string indices
$i,j$.
Let $\oplus$ denote coordinatewise addition modulo two.
Choose a fixed small-bias
family $S$ of pairs of index
translations, and replace a local predicate by
\[
 F'(i,j)=\bigvee_{(s,s')\in S}
              F(u_{i\oplus s},v_{j\oplus s'}).
\]
The left encoder reads only the translated entries of its own
list, and the right encoder does the same. The OR can be
constructed without evaluating all cross-list pairs.

If no original pair is good, no amplified pair is good.
If at least half are good, padding leaves constant density.
The small-bias hitting bound of Ben-Sasson and
Viola~\cite{BenSassonViola2014} bounds the missed fraction by a
constant times the square of the bias. Bias of order $q^{-1/2}$
suffices. The explicit construction of Alon et
al.~\cite{AlonEtAl1992} gives
$|S|=O(q\log^2 p)$.
The amplified encodings have run bound
$r'=O(qr\log^2 p)$. Taking $H=\Theta(qr')$ in the capped list
construction gives equal-length output strings of total length
and construction time
\[
 O(pq(r')^2)=\widetilde O(pr^2q^3).
\]
One factor of $q$ comes from choosing $H=\Theta(qr')$. The other
two come from squaring $r'$, which grows linearly with $q$ up to
logarithmic factors.

\paragraph{Preserving the running-time saving.}
For positive constants $\gamma,a$, take $q=\lceil p^\gamma\rceil$
and require $r\le p^a$.
The common output length satisfies
$N\le p^{1+3\gamma+2a+o(1)}$.
A hypothetical $N^{\kappa\delta}$-approximation in
$O(N^{2-\delta})$ time would solve the source gap in strongly
subquadratic time if
\[
 \kappa\delta(1+3\gamma+2a)<\gamma,\qquad
 (2-\delta)(1+3\gamma+2a)<2.
\]
The first inequality makes its approximation factor smaller than
$q$. The second keeps both construction and approximation time
subquadratic in $p$, since $2-\delta>1$.
With $a=0$, a positive $\gamma$ satisfies both inequalities exactly
when $6\kappa<1$. The constant six comes from the quadratic source
bound and the cubic dependence on $q$.
Choose such a $\gamma$, then a sufficiently small positive $a$.
The matrix encoding has $r=O(KL)$.
Taking $0<\alpha<a$ in the GBD hypothesis gives
$r=O(KL)\le p^a$ for sufficiently large $p$.
This yields Theorem~\ref{thm:explicit}.

For fan-in-two Boolean formulas, we use an OR of two ANDs at every
level, filling unused inputs with constants. This fixed layout
represents both gate types and gives all formulas with the same depth
bound the same collapses on each side. The run bound is $2^{O(d)}$
at depth $d$. Unfolding a
polynomial-size fan-in-two circuit of depth $o(v)$ on $v$ inputs
gives run bound $2^{o(v)}$.
Split the variables into two groups of nearly equal size and
encode their partial assignments separately. Each list has
$p=2^{v/2+O(1)}$ entries, so $r=2^{o(v)}=p^{o(1)}$.
The hypothetical DTW algorithm above then gives a $2^{(1-\Omega(1))v}$-time algorithm distinguishing no
satisfying assignment from a constant fraction of satisfying
assignments. The circuit criterion of Ben-Sasson and
Viola~\cite{BenSassonViola2014} gives the circuit consequence
described in the introduction.

\subsection{Run-length encoded strings}\label{sec:overview-rle}

For RLE inputs, increasing $H$ does not increase the number of
runs. We use this freedom to obtain the gap
\[
 \textsc{Yes}:\ \DTW(X,Y)\le n+m-1,\qquad
 \textsc{No}:\ \DTW(X,Y)\ge H,
\]
with $n,m$ independent of $H$.
The difficult part is to keep the accepting cost independent of
$H$ while forcing every path below $H$ to certify a solution.

Boneh et al.~\cite{BonehEtAl2024} compute
$\min\{\DTW(X,Y),k+1\}$ in $\widetilde O((n+m)k)$ time for a
positive integer threshold $k$. In our instances,
each source list has $p$ vectors of dimension $O(\log p)$, and the
DTW encoding has $n+m=p^{1+o(1)}$ runs. An accepting cost of at most
$(n+m)^{1-\varepsilon}$ for a fixed $\varepsilon>0$ would
let this algorithm decide the gap in strongly subquadratic
time. We obtain the approximation gap by increasing the rejecting
threshold while keeping the accepting cost linear in the run count.

The reduction passes through regular-expression membership. A pattern
$P=D_1\cdots D_s$ is a concatenation of factors. Each required
factor emits one symbol from a specified set; an optional factor
may instead emit the empty word $\epsilon$. A derivation chooses
one permitted emission from each factor, and $L(P)$ consists of
the words emitted by these derivations. A derivation is accepting
for $T$ if it emits $T$. A factor is empty in a derivation if it
emits $\epsilon$. We ask whether a given
text $T=v_1\cdots v_{L_T}$ belongs to $L(P)$.
The source uses a fixed alphabet and eight fixed factor types.

\paragraph{Encoding membership.}
Give each text symbol $v$ a constant-length alternating code
$\operatorname{code}(v)$ over $\{\str{a},\str{b}\}$, and expand
every code position into a heavy
run. Call the resulting block $T_v$.
Encode each factor by a constant-length alternating word in which
some positions are heavy and the others are light.
We choose the factor words so that a factor can emit $v$ exactly
when a subsequence of its word that retains every heavy position
collapses to
$\operatorname{code}(v)$.

For example, the symbol $\str{U}$ has code $\str{aba}$, so
$T_{\str{U}}=\str{a}^H\str{b}^H\str{a}^H$.
The optional $\str{U}$ factor uses
$\str{a}^H\str{b}\str{a}^H$.
It can match $T_{\str{U}}$ at zero cost. It can also match a
light dummy $\str{a}$ at cost one, representing the empty emission.
Every heavy run in an optional factor has label $\str{a}$.
Holding a dummy $\str{a}$ while traversing such a factor costs
one for each light $\str{b}$ position.
Every required factor has a heavy $\str{b}$-run, which cannot
receive an equality cell against such a dummy.

Let $F_i$ denote the factor blocks. Separate consecutive blocks
on both axes by $\str{c}^H$, and put light dummy intervals between
the text blocks. Figure~\ref{fig:rle-overview} shows an accepting
alignment. On \textsc{Yes} instances, the source guarantees an accepting
derivation with at most $\ell$ consecutive empty factors in each
internal text gap and at most $L_T$ at each end. Place $\ell$
dummies in each internal gap. We also place $L_T$ dummies at each
end to handle empty prefix and suffix factors.

\Needspace{18\baselineskip}
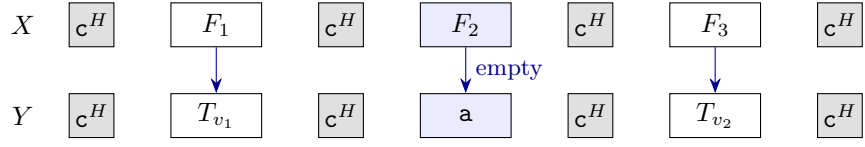
\begin{figure}[htbp]
\centering
\begin{tikzpicture}[x=1.65cm,y=1cm,font=\small,
 profile/.style={draw,minimum width=1.2cm,minimum height=.58cm,inner sep=3pt},
 guard/.style={draw,fill=black!12,minimum width=.55cm,minimum height=.58cm,inner sep=2pt},
 match/.style={-{Stealth[length=2mm]},blue!55!black}]
\node at (-.55,.9) {$X$};
\node at (-.55,-.3) {$Y$};
\foreach \x in {0,2,4,6} {
  \node[guard] at (\x,.9) {$\str{c}^H$};
  \node[guard] at (\x,-.3) {$\str{c}^H$};
}
\node[profile] (f1) at (1,.9) {$F_1$};
\node[profile,fill=blue!8] (f2) at (3,.9) {$F_2$};
\node[profile] (f3) at (5,.9) {$F_3$};
\node[profile] (t1) at (1,-.3) {$T_{v_1}$};
\node[profile,fill=blue!8] (dum) at (3,-.3) {$\str{a}$};
\node[profile] (t2) at (5,-.3) {$T_{v_2}$};
\draw[match] (f1.south)--(t1.north);
\draw[match] (f2.south)--node[right,font=\footnotesize] {empty} (dum.north);
\draw[match] (f3.south)--(t2.north);
\end{tikzpicture}
\caption{Factor gadgets $F_1,F_2,F_3$ in an accepting alignment.
The optional factor $F_2$ emits no symbol and uses a light dummy $\str{a}$;
the other factors emit the symbols represented by $T_{v_1},T_{v_2}$.
Heavy separators force an ordered matching between factors and
text or dummy intervals in every path of cost below $H$.
Other dummy intervals and their separators are omitted.}
\label{fig:rle-overview}
\end{figure}

A path below cost $H$ gives every fully traversed heavy run an
equality cell. The heavy separators prevent one factor from
having equality cells in two text or dummy intervals: the path
would cross a full $\str{c}^H$ while remaining in a block containing
no $\str{c}$. The same argument prevents one interval from meeting
two factors. Equality cells define an ordered matching.
It covers every factor and every real text interval.

Within a matched factor and real interval, retain the factor
positions that have equality cells. Every heavy factor position
is retained, and every heavy code run has a witness. One factor run
cannot have equality cells in two code runs: traversing the
opposite-label heavy run between them would cost at least $H$.
Monotonicity orders the retained groups, whose run collapse is
the whole code. Hence the factor can emit that text symbol.
Factors matched to dummies are optional and emit $\epsilon$.
The ordered matching gives an accepting derivation,
proving the \textsc{No} bound.

For the \textsc{Yes} bound, start with an accepting derivation and work on
the run collapses. Match each factor to its emitted symbol or to
a distinct dummy if it is empty. For each factor that emits $v$,
choose a subsequence of its word that retains every heavy position
and collapses to $\operatorname{code}(v)$. Match each run of this
subsequence to the corresponding code position. Visit deleted
light positions while holding the preceding code position, or
the first one before any retained position is reached.
Between rectangles, hold a left $\str{c}$ while traversing
right separators and unused dummies.
This gives a full path with an equality cell for every position
designated as heavy. Its cost is at most $n+m-1$, the maximum
number of cells in a path on the run collapses.
Expanding the positions designated as heavy to length $H$ adds
no cost.

\Needspace{12\baselineskip}
\paragraph{The Orthogonal Vectors source.}
Orthogonal Vectors (OV) asks whether two lists of $p$ vectors in
$\{0,1\}^d$ contain a pair, one from each list, with inner product
zero. We use the Orthogonal Vectors Hypothesis (OVH).

\begin{hypothesis}\label{hyp:rovh}
For every $\varepsilon>0$, there is a constant $c>0$ such that no
bounded-error randomized algorithm solves Orthogonal Vectors on two lists of $p$
vectors in dimension $d=\lceil c\log p\rceil$ in
$O(p^{2-\varepsilon})$ time.
\end{hypothesis}

The standard reduction from SAT to OV~\cite{Williams2005}, combined
with sparsification~\cite{ImpagliazzoPaturiZane2001}, shows that
randomized SETH~\cite{BennettPeikertTang2022} implies OVH.

We adapt the membership reduction of Backurs and
Indyk~\cite{BackursIndyk2016} to the eight factor types.
The text contains blocks encoding the right vectors,
interleaved with neutral blocks and markers. The pattern contains
a gadget for each left vector. That gadget can emit a right vector
block exactly when the two vectors are orthogonal, and it can
always emit a neutral block.

Auxiliary gadgets before and after each left vector gadget emit
strings whose lengths range from one text block to two. A fixed
gadget for the all-one vector follows each trailing auxiliary gadget.
This fixed gadget cannot emit a marker block. All vector and auxiliary
gadgets begin with $\str{U}$.

If every auxiliary gadget uses its default two-block emission,
each left vector gadget reads a neutral block. This alignment
cannot reach the final marker. In an accepting derivation, consider
the first auxiliary gadget that uses a shorter emission.
An intermediate length places the following gadget at a non-$\str{U}$
text symbol. The first shorter emission must span exactly one block.
If this auxiliary gadget followed a left vector gadget, the fixed
all-one gadget would have to emit a marker block. The first shorter choice
must occur immediately before
a left vector gadget, which must then emit a right vector block.
This certifies an orthogonal pair. Conversely, an orthogonal pair
permits that choice and the later shift needed to reach the final
marker.

Optional padding converts substring membership into membership
of the whole text. The source construction has
$L_T,s=\Theta(pd)$ and allows $\ell=\Theta(d)$.
The factor string has $n=\Theta(pd)$ runs. In the text string,
the $\ell$ dummy intervals in each gap give
$m=\Theta(pd^2)$ runs. Their product is
\[
 nm=\Theta(p^2d^3).
\]
OVH supplies a hard dimension $d=\lceil c\log p\rceil$ for a suitable
constant $c$. Construction then takes
$p^{1+o(1)}$ time, whereas an
$\widetilde O((nm)^{1-\delta})$-time approximation would give a
strongly subquadratic OV algorithm.

For a constant approximation factor $C$, choose
$H=\Theta_C(n+m)$ with $H>C(n+m)$.
A successful estimate is below $H$ on a \textsc{Yes} instance
and at least $H$ on a \textsc{No} instance, proving
Theorem~\ref{thm:rle}(\ref{thm:rle-constant}).
For each fixed $\eta\in(0,1)$, choose
$H=(n+m)^{\lceil2/\eta\rceil}$.
Since $N+M\le(n+m)H$, this choice gives
$(N+M)^{1-\eta}(n+m-1)<H$.
The run count is unchanged, and $\log(N+M)=O_\eta(\log p)$.
The same running-time contradiction proves
Theorem~\ref{thm:rle}(\ref{thm:rle-large}).

\Needspace{8\baselineskip}
\section{Preliminaries}\label{sec:preliminaries}

All logarithms are to base two, and $[k]=\{1,\ldots,k\}$.
Concatenation is written by juxtaposition, and $\bigodot$ denotes
indexed concatenation in increasing index order. For a symbol $\sigma$,
$\sigma^k$ denotes $k$ copies of $\sigma$. We use the standard
word-RAM model, with
$\Theta(\log(N+M))$-bit words for DTW.

A $C$-approximation to a distance $D$ returns a value $\widehat D$
satisfying $D\le\widehat D\le CD$.
Randomized approximation algorithms must satisfy this guarantee
with probability at least $2/3$ on every input.

\subsection{Run notation and fixed metrics}

Write the run decompositions as
\[
 X=x_1^{\ell_1}\cdots x_n^{\ell_n},\qquad
 Y=y_1^{\ell'_1}\cdots y_m^{\ell'_m},
\]
where successive run labels differ and all run lengths are positive.
The expanded lengths are
\[
 N=\sum_{i=1}^n\ell_i,\qquad M=\sum_{j=1}^m\ell'_j.
\]
A run-length encoding stores each run's symbol and length, with
the length written in binary.

\begin{remark}\label{rem:fixed-metrics}
Let $\rho$ be any fixed metric on $\Sigma$. Normalize its minimum
nonzero distance to one, and let $\Delta$ be its maximum distance.
Write $\DTW_\rho$ for the minimum path cost when a cell $(i,j)$ costs
$\rho(X[i],Y[j])$. Every path satisfies the corresponding cost bounds,
so
\[
  \DTW(X,Y)\le\DTW_\rho(X,Y)\le\Delta\DTW(X,Y).
\]
A $C$-approximation for $\DTW_\rho$ is a
$\Delta C$-approximation for uniform-metric DTW on the same strings.
For fixed $\kappa\in(0,1/6)$ and $\delta\in(0,1)$, choose
$\kappa'$ strictly between $\kappa$ and $1/6$. Then
$\Delta N^{\kappa\delta}\le N^{\kappa'\delta}$ for sufficiently
large $N$, so Theorem~\ref{thm:explicit} gives the same exponent
bound for $\rho$. Theorem~\ref{thm:rle}(\ref{thm:rle-constant})
applies with constant factor
$\Delta C$, preserving its run-length restrictions with constants
that may also depend on $\rho$. Likewise,
$\Delta (N+M)^{1-\eta}\le (N+M)^{1-\eta/2}$ for sufficiently large
$N+M$,
so Theorem~\ref{thm:rle}(\ref{thm:rle-large}) extends to $\rho$.
For absolute distance on $\{0,1,2\}$, we have $\Delta=2$.
Restricting to three points gives the same lower bounds in any fixed
metric with at least three points.

For any fixed finite metric,
multiplying a uniform-metric estimate by $\Delta$ gives an estimate
for that metric with a constant-factor loss. This transfers
Kuszmaul's deterministic approximation algorithm. The comparison of
approximation exponents holds within any fixed factor greater than six
for every fixed finite metric with at least three points.
\end{remark}

\Needspace{16\baselineskip}
\subsection{Basic facts}

\begin{lemma}\label{gate:path-facts}
For nonempty strings $U,V$, the following properties hold.
\begin{enumerate}
\item $\DTW(U,V)=0$ if and only if $\shrink(U)=\shrink(V)$.
\item Every path on $U,V$ has at most $|U|+|V|-1$ cells, and
$\DTW(U,V)\le \max\{|U|,|V|\}$.
\item Increasing the lengths of specified runs cannot decrease DTW.
\item If $U=U_1\cdots U_s$ and $V=V_1\cdots V_s$, where every
displayed block is nonempty, then
\[
  \DTW(U,V)\le\sum_{i=1}^{s}\DTW(U_i,V_i).
\]
\item If $U,V$ end in the same symbol, extending the final run of
either string preserves their distance exactly.
\end{enumerate}
\end{lemma}

\begin{proof}
If the run collapses agree, align the corresponding constant runs
at zero cost and concatenate these alignments.  Conversely, on a
zero-cost path, leaving a pair of corresponding constant runs
requires changing the run index on both axes at the same step.
Changing just one run index would compare distinct symbols.  Every
run on either axis is visited, so the two sequences of run labels
coincide.

Each step increases the sum of the coordinates, so any path has at
most $|U|+|V|-1$ cells. Use diagonal steps until one coordinate
reaches its last position, then advance only the other coordinate.
This path has $\max\{|U|,|V|\}$ cells.

For the third assertion, contract each enlarged run to its original
length. Map its positions monotonically onto the original positions,
with increments zero or one, preserving the two endpoints. Apply
these maps on both axes of a path in the enlarged strings. Each
projected step is either a valid warping-path step or stays at the
same cell. Removing consecutive repeated cells gives a
valid path in the original strings. The symbols at the projected
positions are unchanged, and removing repetitions cannot increase
the cost.

For the fourth assertion, concatenate optimal block paths. The
last cell for block pair $i$ and the first cell for block pair $i+1$
differ by $(1,1)$. The resulting path contains exactly the cells of
the block paths, so its cost is their sum.

For the fifth assertion, an old path extends through added copies of the common
terminal symbol at zero cost.  The reverse inequality follows from
the run-contraction argument above.
\end{proof}

\Needspace{12\baselineskip}
\begin{lemma}\label{gate:heavy-witnesses}
\begin{enumerate}[label=(\roman*)]
\item Let $\pi$ be a path segment of cost less than $H$.
Suppose that $\pi$ visits every position of heavy runs
$R_1,\ldots,R_k$ on one axis, listed in their order in the string.
Each $R_i$ has an equality cell on $\pi$.
If consecutive $R_i$ have different labels, then any choice of one
equality cell from each run gives strictly increasing run indices
on the opposite axis.
\item Let $U,V$ be run-collapsed strings, and select positions
in the two strings to expand into heavy runs. If a path on $U,V$
contains an equality cell for every selected position, then
replacing each selected position by $H$ copies of its symbol gives
expanded strings with a path of the same cost.
\end{enumerate}
\end{lemma}

\begin{proof}
For (i), suppose that some $R_i$ has no equality cell on $\pi$.
The path visits every one of its $H$ positions, so each position
occurs in a mismatching cell. These give $H$ distinct unit-cost
cells because their coordinates on that axis differ. This
contradicts the cost bound.

Choose one equality cell in each $R_i$. Since the runs occur in
string order, the chosen cells occur in path order. Monotonicity
makes their run indices on the opposite axis nondecreasing.
If consecutive $R_i$ have different labels, their chosen cells
cannot lie in the same opposite run, whose positions all have the
same symbol. The run indices increase strictly.

For (ii), expand one selected position at an equality cell.
Insert $H-1$ additional copies of its symbol and traverse them
while holding the opposite coordinate fixed. Every added cell
compares equal symbols and has cost zero. The existing comparisons
keep their costs, and the equality cells for the other selected
positions remain available. Repeat this operation for all selected
positions in both strings. The resulting path has the same cost
as the original path.
\end{proof}

\section{Hardness for explicit strings}\label{sec:explicit}

We first encode individual predicate comparisons with Boolean gadgets.
We then amplify the fraction of true pairs, cap the distance of every
comparison, and combine the lists with heavy guards. The guarded
strips turn false comparisons into an additive lower bound. We prove
the cap and list construction before amplification, which uses both.

\subsection{Boolean gadgets}\label{gate:section}

The endpoint labels below let us recover a complete child comparison
from a path segment that traverses both children. This is needed
when tags isolate a singleton slot but the segment starts or ends
outside the child rectangle.

\begin{lemma}\label{gate:endpoint-isolation}
Let $\pi$ be a path segment in the grid for $X^-XX^+$ and $Y^-YY^+$.
Suppose that $\pi$ has cost below $H$ and visits every position of
$X$ and of $Y$. Assume that $X,Y$ begin with $\str{c}^H$ and end with
$\str{b}^H$. Assume also that the prefixes $X^-,Y^-$ contain no
$\str{c}$ and the suffixes $X^+,Y^+$ contain no $\str{b}$.
Then $\DTW(X,Y)<H$.
\end{lemma}

\begin{proof}
Both initial heavy runs have equality cells on $\pi$, by
Lemma~\ref{gate:heavy-witnesses}. Suppose no equality cell joins
these two runs. An equality cell $(i,j)$ for the initial run of $X$
must have its right coordinate after the initial run of $Y$:
it cannot be in $Y^-$, which contains no $\str{c}$. Similarly, an
equality cell $(i',j')$ for the initial run of $Y$ must have its left
coordinate after the initial run of $X$. Thus $i<i'$ but $j>j'$.
A monotone path cannot contain both cells. The two initial
$\str{c}$-runs must have an equality cell with each other.

The same argument at the terminal runs, using the absence of
$\str{b}$ in $X^+,Y^+$, gives an equality cell between the two
terminal $\str{b}$-runs. By monotonicity, the subpath between these
two cells stays in the child rectangle. Extend it to the rectangle's
first corner through the initial $\str{c}$-runs and to its last
corner through the terminal $\str{b}$-runs. These extensions cost
zero, so the resulting full path on $X,Y$ has cost at most the cost
of $\pi$, which is below $H$.
\end{proof}

\subsubsection{Tags and ordered slot intervals}

Fix a child run bound $r\ge3$ and put
\[
 h=r+2,\qquad P=(\str{a}^H\str{b}^H)^h,\qquad Q=(\str{a}^H\str{c}^H)^h,
 \qquad t=4h+r=5r+8.
\]
Throughout the gate constructions, every child has at most $r$ runs,
all of length one or $H$, and endpoints $\str{c}^H,\str{b}^H$.
For a child $X$, its \emph{active wrapper} is $\widehat X=PXQ$.
The run collapse of an active wrapper is a \emph{light wrapper}.
Each wrapper has at most
$t$ runs. A slot contains one active or light wrapper.
The child inside a light wrapper is the run collapse of the
corresponding active child, so it also has at most $r$ runs.
The displayed $P,Q$ tags at this construction level, including their
light versions, are called \emph{genuine}; occurrences inside children
are not genuine tags.

Child endpoints differ from the adjacent tag labels.
Wrappers begin with $\str{a}$ and end with $\str{c}$, so no runs merge
at the child, slot, or outer-run boundaries.

\begin{lemma}\label{gate:slot-spans}
Consider a left string consisting of an initial $\str{c}^H$, a sequence of
slots with no two consecutive light slots, and a terminal $\str{b}^H$.
The right string consists of the same outer runs and a sequence of
active or light wrappers.
Every path of cost below $H$ has the following properties.
\begin{enumerate}[label=(\roman*)]
\item Every genuine heavy $P$-tag on the right has a $\str{b}/\str{b}$
equality cell in a genuine $P$-tag on the left.
Every genuine heavy $Q$-tag on the right has a $\str{c}/\str{c}$
equality cell in a genuine $Q$-tag on the left.
\item For each active right wrapper formed by genuine tags, choose
these leading and trailing equality cells as anchors. The left
slots from the leading anchor's slot to the trailing anchor's slot,
including both, form a nonempty interval. The intervals for
successive active right wrappers are ordered and disjoint.
\item If the interval for an active right wrapper with child $Y_j$
consists of a single active left slot with child $X_i$, then
$\DTW(X_i,Y_j)<H$.
\end{enumerate}
\end{lemma}

\begin{proof}
For the leading tags, we show that a right $P$-tag cannot obtain all
its equality witnesses from children and the outer run. By
Lemma~\ref{gate:heavy-witnesses}, its alternating heavy runs have
witnesses with strictly increasing left run indices. In particular,
its $h$ $\str{b}$-runs require $h$ distinct left $\str{b}$-runs.
Suppose the right $P$ has no $\str{b}/\str{b}$ equality cell in
any genuine $P$-tag on the left. Its witnesses must then lie in children or
in the outer terminal run.

Witnesses in two different children would force the path to cross
a slot boundary while its right coordinate stays in this $P$.
If the earlier slot is active, the path must traverse a heavy
$\str{c}$-run of its trailing $Q$. Every cell in that traversal
mismatches because $P$ contains no $\str{c}$. If the earlier slot
is light, the next slot is active. The path must then traverse a
heavy $\str{b}$-run of that slot's leading $P$. This run has no equality
cell with the right $P$, by our assumption. Either traversal costs
at least $H$. All child witnesses must lie in one child.
That child has at most $r$ runs. Together with the outer terminal run, this gives
at most $r+1<h$ distinct left $\str{b}$-runs, a contradiction.
Figure~\ref{fig:tag-crossing} shows the crossing when the earlier
slot is active.

For a genuine heavy $Q$-tag on the right, suppose it has no $\str{c}/\str{c}$
equality cell in any genuine $Q$-tag on the left. Crossing an active slot's trailing $Q$
would traverse a heavy $\str{c}$-run with no equality cell by this
assumption. Crossing a light slot's next active $P$ would traverse
a heavy $\str{b}$-run, which always mismatches the right $Q$.
The witnesses must again lie in one child, together with possibly
the initial outer run. These cannot supply the $h$ distinct
$\str{c}$-runs required by the right tag.

The leading anchor precedes the trailing anchor on the right.
The anchors occur in the same order on the path.
Their left slot indices are nondecreasing,
so they define a nonempty interval. For two successive active right
wrappers, the first wrapper's trailing anchor precedes the second
wrapper's leading anchor. If both anchors were in the same left
slot, the path would visit that slot's $Q$ before its $P$. This is
impossible because $P$ precedes $Q$ on the left.
The first wrapper's trailing anchor lies in an earlier left slot
than the second wrapper's leading anchor. The intervals are ordered
and disjoint.

For a singleton interval containing an active left child $X_i$,
the subpath between the anchors traverses both $X_i$ and the right
child $Y_j$ completely. On each axis, the leading fragment lies
in $P$ and contains no $\str{c}$, and the trailing fragment lies
in $Q$ and contains no $\str{b}$. Lemma~\ref{gate:endpoint-isolation}
gives $\DTW(X_i,Y_j)<H$.
\end{proof}

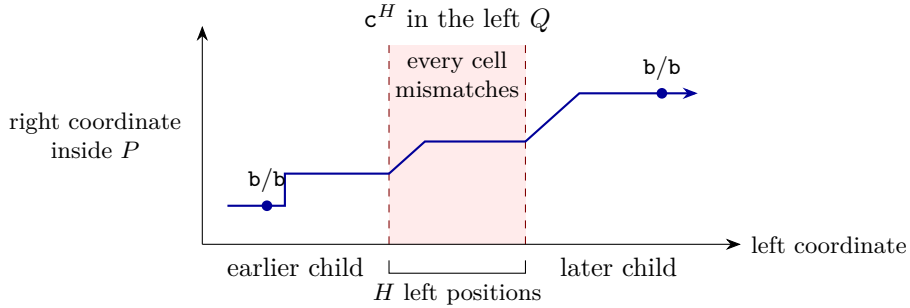
\begin{figure}[htbp]
\centering
\begin{tikzpicture}[x=.95cm,y=.85cm,font=\small,
 pathline/.style={thick,blue!60!black,-{Stealth[length=2mm]}}]
\fill[red!8] (2.6,0) rectangle (4.5,3.1);
\draw[dashed,red!50!black] (2.6,0)--(2.6,3.1);
\draw[dashed,red!50!black] (4.5,0)--(4.5,3.1);
\draw[-{Stealth[length=2mm]}] (0,0)--(7.5,0)
 node[right,font=\footnotesize] {left coordinate};
\draw[-{Stealth[length=2mm]}] (0,0)--(0,3.4);
\node[left,align=center,font=\footnotesize] at (-.15,1.7)
 {right coordinate\\inside $P$};
\draw[pathline] (.35,.6)--(1.15,.6)--(1.15,1.1)--(2.6,1.1)
 --(3.1,1.6)--(4.5,1.6)--(5.25,2.35)--(6.9,2.35);
\fill[blue!60!black] (.9,.6) circle (2pt);
\fill[blue!60!black] (6.4,2.35) circle (2pt);
\node[above,font=\footnotesize] at (.9,.65) {$\str{b}/\str{b}$};
\node[above,font=\footnotesize] at (6.4,2.4) {$\str{b}/\str{b}$};
\node[align=center,font=\footnotesize] at (3.55,2.65)
 {every cell\\mismatches};
\node[above] at (3.55,3.15) {$\str{c}^H$ in the left $Q$};
\node[below] at (1.3,-.05) {earlier child};
\node[below] at (5.8,-.05) {later child};
\draw (2.6,-.25)--(2.6,-.45)--(4.5,-.45)--(4.5,-.25);
\node[below,font=\footnotesize] at (3.55,-.45) {$H$ left positions};
\end{tikzpicture}
\caption{An attempted passage between children while the right
coordinate stays in a leading tag $P$. If the earlier left slot is
active, the path crosses a heavy $\str{c}$-run in its trailing tag $Q$.
Since $P$ contains no $\str{c}$, all cells in the shaded band mismatch.
The path visits all $H$ left positions of this run and pays at least $H$.
Only the relevant segment is shown.}
\label{fig:tag-crossing}
\end{figure}

\Needspace{16\baselineskip}
\subsubsection{AND and OR}

The AND construction concatenates the wrapped children in the same
order on both sides.
\begin{proposition}\label{gate:and}
For $k\ge1$ child pairs satisfying the preceding endpoint and run
assumptions, set
\[
 X_\wedge=\str{c}^H\widehat X_1\cdots\widehat X_k \str{b}^H,
 \qquad
 Y_\wedge=\str{c}^H\widehat Y_1\cdots\widehat Y_k \str{b}^H.
\]
Then
\[
 \DTW(X_\wedge,Y_\wedge)\le\sum_{i=1}^k\DTW(X_i,Y_i),
 \qquad
 \bigl(\exists i:\DTW(X_i,Y_i)\ge H\bigr)
 \Longrightarrow\DTW(X_\wedge,Y_\wedge)\ge H.
\]
Each output has at most $kt+2$ runs, all of length one or $H$, and
endpoints $\str{c}^H,\str{b}^H$.
\end{proposition}

\begin{proof}
Match corresponding tags and outer runs at zero cost and concatenate
the child paths, using Lemma~\ref{gate:path-facts}.
For soundness, a path below $H$ would assign the $k$ right wrappers
$k$ ordered disjoint nonempty intervals among exactly $k$ left slots,
by Lemma~\ref{gate:slot-spans}. Every interval is its corresponding
singleton, so every corresponding child pair has distance below $H$.
The remaining assertions follow by counting wrappers and the
distinct-label joins above.
\end{proof}

\paragraph{OR.}
Assume in addition that all left children have the same fixed
collapse $A$, and all right children have the same fixed collapse $B$.
Define the light wrappers
\[
 U_X=\shrink(PX_iQ),\qquad U_Y=\shrink(PY_iQ).
\]
Their lengths are at most $t$, and they do not depend on the child
index or the encoded data.
In the layout below, every unselected active wrapper can match a
dummy on the opposite side with the same collapse. Selecting one
corresponding child pair leaves two right dummies to traverse while
holding the left coordinate fixed.
Using run collapses as zero-cost dummy gadgets also appears in the
reduction of Boneh, Golan, and Kociumaka~\cite{BonehGolanKociumaka2026}.

\Needspace{16\baselineskip}
\begin{proposition}\label{gate:or}
For $k\ge1$, define
\[
\begin{aligned}
 X_\vee&=\str{c}^H\widehat X_1U_Y\widehat X_2U_Y\cdots
             U_Y\widehat X_k \str{b}^H,\\
 Y_\vee&=\str{c}^HU_X\widehat Y_1U_X\widehat Y_2U_X\cdots
             \widehat Y_kU_X \str{b}^H.
\end{aligned}
\]
Then
\[
 \DTW(X_\vee,Y_\vee)\le\min_i\DTW(X_i,Y_i)+2t,
 \qquad
 \bigl(\forall i:\DTW(X_i,Y_i)\ge H\bigr)
 \Longrightarrow\DTW(X_\vee,Y_\vee)\ge H.
\]
The left and right outputs have at most $(2k-1)t+2$ and $(2k+1)t+2$
runs, respectively. Every run has length one or $H$.
Both outputs begin with $\str{c}^H$ and end with $\str{b}^H$.
\end{proposition}

\begin{proof}
Fix a selected child index $i$. The dummies let us match every
other active wrapper at zero cost. The table gives the order of
comparisons.
\[
\begin{array}{c|cc}
\text{child index}&\text{first comparison}&\text{second comparison}\\
j<i&\widehat X_j\leftrightarrow U_X&U_Y\leftrightarrow\widehat Y_j\\
j>i&U_Y\leftrightarrow\widehat Y_j&\widehat X_j\leftrightarrow U_X
\end{array}
\]
Before $i$, use the right dummy preceding $\widehat Y_j$ and the
left dummy following $\widehat X_j$. After $i$, use the left dummy
preceding $\widehat X_j$ and the right dummy following
$\widehat Y_j$. Each comparison has equal run collapses.

At the selected index, traverse the preceding right $U_X$ while
holding the first position of $\widehat X_i$. Compare
$\widehat X_i$ with $\widehat Y_i$ by matching their tags at zero
cost and using an optimal path on $X_i,Y_i$. Then traverse the
following right $U_X$ while holding the last position of
$\widehat X_i$. These are the only dummy traversals that may incur
cost, and together they cost at most $2|U_X|$.

To join these pieces, enter the first dummy traversal by a diagonal
step, then advance only the right coordinate into the selected
wrapper. After the selected comparison, advance only the right
coordinate into the second dummy and leave it by a diagonal step.
All other joins concatenate complete block paths, including the
outer runs. The outer runs provide the adjacent blocks when the
selected child is first or last. The total cost is at
most $\DTW(X_i,Y_i)+2|U_X|\le\DTW(X_i,Y_i)+2t$.

For soundness, Lemma~\ref{gate:slot-spans} assigns a path below $H$
ordered disjoint nonempty intervals $I_j=[a_j,b_j]$, $1\le j\le k$,
among the $2k-1$ left slots. Slot $2j-1$ contains $X_j$.
We show that some interval is exactly $\{2j-1\}$, as illustrated
in Figure~\ref{fig:or-slots}. If no corresponding singleton exists,
then $b_1\ge2$.
Suppose inductively that $b_{j-1}\ge2j-2$. Disjointness and order
give $a_j\ge b_{j-1}+1\ge2j-1$. If $b_j=2j-1$, the interval
would be the excluded singleton. Hence $b_j\ge2j$. At $j=k$
this contradicts $b_k\le2k-1$. The resulting singleton interval
gives $\DTW(X_j,Y_j)<H$ by Lemma~\ref{gate:slot-spans}, proving
soundness. Counting the wrappers and using the distinct-label
joins gives the run and endpoint assertions.
\end{proof}

\begin{lemma}\label{gate:collapse-preservation}
The AND and OR constructions have the following collapse properties.
\begin{enumerate}[label=(\roman*)]
\item If each child's collapses are independent of the encoded
data, then the AND output collapses are independent of the data.
The child collapses may depend on the child index. The same holds
for OR under its common-collapse hypothesis. In both constructions,
if the child collapses are independent of $H$, so are the output
collapses.
\item If the AND children have common left and right collapses
$A\ne B$, its output collapses are distinct. The OR output
collapses are always distinct.
\end{enumerate}
\end{lemma}

\begin{proof}
The tags and outer runs have fixed collapses, and no joins merge
runs. The output collapses are concatenations of these
fixed strings and the child collapses. This proves the independence
assertions.

For common child collapses $A,B$, write
$C_A=(\str{ab})^hA(\str{ac})^h$ and $C_B=(\str{ab})^hB(\str{ac})^h$.
The AND collapses are $\str{c}C_A^k\str{b}$ and $\str{c}C_B^k\str{b}$.
Equality of the output collapses forces $|C_A|=|C_B|$.
The first blocks then give $C_A=C_B$.
Removing the common tags gives $A=B$.
The right OR collapse is longer than the left by $|C_A|+|C_B|>0$.
\end{proof}

The run bounds depend only on the child run bound and the number of
children. Once the children and fixed dummies are available, each gate
output is written in time linear in its expanded length.

\subsection{Encoding the matrix predicate}\label{sec:encoding}

\subsubsection{Standardized literal encodings}\label{sec:literal-encoders}

For a literal whose truth value $z$ is known to the left encoder, set
\[
 \ell(z)=1+(H-1)(1-z),\qquad X(z)=\str{c}^H\str{a}^{\ell(z)}\str{b}^H,
 \qquad Y=\str{c}^H\str{b}^H.
\]
For a literal known to the right encoder, exchange the strings.
A true literal has distance
at most one. A false literal has distance at least $H$, because all
positions of its middle $\str{a}$-run mismatch.

These collapses depend on which side owns the literal. Standardize
them with a two-input AND whose first slot is reserved for a left
literal and whose second is reserved for a right literal. Put the
given literal in its reserved slot and a true literal in the other.
The left child collapses in slot order are $\str{cab},\str{cb}$; the right child
collapses are $\str{cb},\str{cab}$. The output collapse on each side is fixed,
independently of the literal's owner or value. The two side collapses
differ inside the first child, immediately after its initial $\str{c}$.
True literals have distance at most two, and false literals have
distance at least $H$. The run bound is $r_0=2(5\cdot3+8)+2=48$.

\subsubsection{From literals to matrices}

\begin{proposition}\label{prop:predicate-encoders}
For every $K,L\ge1$, the Gap Block Disjointness predicate
\[
 F(A,B)=\bigvee_{k=1}^K\bigwedge_{j=1}^L
                (\neg A_{k,j}\vee\neg B_{k,j})
\]
has separate encodings with run bound $r=O(KL)$.
\end{proposition}

\begin{proof}
Use an OR of the two standardized literals at each coordinate, an
AND of the $L$ coordinates in each row, and an OR of the $K$ rows.
Lemma~\ref{gate:collapse-preservation} maintains the common-collapse
condition at both OR levels and distinct side collapses throughout.
Let $r_1,r_2,r_3$ bound the runs in a coordinate encoding, a row
encoding, and the final matrix encoding, respectively. The gate
bounds give
\[
 r_1=5(5r_0+8)+2,\qquad
 r_2=L(5r_1+8)+2,\qquad
 r_3=(2K+1)(5r_2+8)+2=O(KL).
\]
At a true AND, sum the child costs; at a true OR, select a true child
and pay the additional $2(5r+8)$. The displayed run bounds also bound
these costs. A false AND has a false child, and every child of a false
OR is false, so Propositions~\ref{gate:and} and~\ref{gate:or} give
distance at least $H$. Both gates preserve the endpoints and allowed
run lengths. There are four construction levels, including literal
standardization, each written in time linear in its output size.
Each encoder uses only its own matrix, giving $O(Hr_3)$ time.
\end{proof}

\subsection{Combining separate encodings}\label{pack:section}

Fix a predicate with separate encodings of run bound $r$.
We cap every local distance at $H+O(r)$ while preserving the
true- and false-pair bounds. To apply OR, the fixed pair that
provides the cap must have the original side collapses.

\subsubsection{Capping local costs}

\begin{lemma}\label{pack:cap}
Let $C_X\ne C_Y$ be collapsed strings of length at most $r$, both
beginning with $\str{c}$ and ending with $\str{b}$.
For every $H\ge2$, there are strings $x_{\mathrm{cap}},y_{\mathrm{cap}}$
with these exact collapses. Every run has length one or $H$,
and both strings begin with $\str{c}^H$ and end with $\str{b}^H$.
Their distance satisfies
\[
 H\le\DTW(x_{\mathrm{cap}},y_{\mathrm{cap}})\le H+r.
\]
\end{lemma}

\begin{proof}
Let $\ell$ be the common-prefix length of $C_X,C_Y$. Make the first
$\min\{\ell+1,|C_X|\}$ left runs and the first
$\min\{\ell+1,|C_Y|\}$ right runs heavy, leaving the rest light.
Finally make both terminal $\str{b}$-runs heavy as well.

Suppose there is a path of cost below $H$. By
Lemma~\ref{gate:heavy-witnesses}, the heavy prefix runs have
equality witnesses with strictly increasing opposite run indices.
If both collapses continue past their common prefix, put
$i=\ell+1$. The witness for left run $i$ must have right run index
at least $i$. It cannot have index $i$, since the two labels there
differ, so it lies after right run $i$. The witness for right run
$i$ similarly lies after left run $i$. The first cell is earlier
on the left axis and later on the right axis than the second cell,
contradicting monotonicity.

If one collapse is a proper prefix of the other, all runs of the
shorter collapse and the first $\ell+1$ runs of the longer collapse
are heavy. The latter require $\ell+1$ distinct witness runs on
the shorter axis, which has only $\ell$ runs. This is also
impossible. The distance is at least $H$.

For the upper bound when both collapses continue past the common
prefix, first match that prefix at zero cost. Compare the next
heavy runs diagonally, paying $H$. Consume the remaining light
tails together. If one axis reaches its terminal $\str{b}$-run
first, hold a $\str{b}$ there while finishing the other light tail.
These tails cost at most $r$ further mismatches. Finish by matching
the terminal $\str{b}$-runs at zero cost.

In the proper-prefix case, match the common prefix and then hold
the shorter string's final $\str{b}$. Traverse the next heavy run
of the longer string and its light tail, paying at most $H+r$,
then finish on the terminal $\str{b}$-runs.
\end{proof}

\begin{lemma}\label{pack:normalized}
Set $r_\ast=100r$. There are separate encodings $A(u),B(v)$ with
run bound $r_\ast$ that also satisfy
\[
 \DTW(A(u),B(v))\le H+r_\ast\qquad\text{for every }u,v.
\]
\end{lemma}

\begin{proof}
Apply Lemma~\ref{pack:cap} to the fixed side collapses. Use the
predicate pair and the resulting cap pair as the two children of
Proposition~\ref{gate:or}. Their side collapses agree.
With $t=5r+8$, its upper bound is
\[
 \min\{\DTW(X_F(u),Y_F(v)),
        \DTW(x_{\mathrm{cap}},y_{\mathrm{cap}})\}+2t.
\]
Selecting the predicate child costs at most $r+2t\le r_\ast$ on a
true pair; selecting the cap child always costs at most $H+r_\ast$.
On a false pair both children have distance at least $H$, so OR
soundness applies. The larger run bound is $5t+2\le r_\ast$.
The remaining properties in Definition~\ref{def:encoders} follow from
the OR construction and Lemma~\ref{gate:collapse-preservation}.
\end{proof}

\subsubsection{Extracting a guarded block}

Define the heavy outer tags and their guard by
\[
 P_{\mathrm{out}}=(\str{a}^H\str{b}^H)^{r_\ast+1},\qquad
 Q_{\mathrm{out}}=(\str{a}^H\str{c}^H)^{r_\ast+1},\qquad
 G^{\mathrm{heavy}}=Q_{\mathrm{out}}P_{\mathrm{out}}.
\]
We place $P_{\mathrm{out}}$ before each block and $Q_{\mathrm{out}}$
after it. The preceding tag contains no $\str{c}$, and the following
tag contains no $\str{b}$, as Lemma~\ref{gate:endpoint-isolation} requires.
A genuine outer tag is one of these displayed guard parts, not an
occurrence inside a block.

\begin{lemma}\label{pack:extraction}
Let $A$ and $B_1,\ldots,B_k$ have at most $r_\ast$ runs each,
begin with $\str{c}^H$, and end with $\str{b}^H$. Let the right string be
\[
 G^{\mathrm{heavy}}B_1G^{\mathrm{heavy}}\cdots
 B_kG^{\mathrm{heavy}}.
\]
Consider a path segment that
visits every left position of $P_{\mathrm{out}}AQ_{\mathrm{out}}$
and no left position outside it.
If the segment costs less than $H$, then $\DTW(A,B_j)<H$ for some $j$.
\end{lemma}

\begin{proof}
We first find equality cells joining the two left outer tags to
right outer tags. We then show that these right tags bound a single
block, which lets us recover a full path on $A$ and that block.

By Lemma~\ref{gate:heavy-witnesses}, the leading left
$P_{\mathrm{out}}$ requires witnesses in $r_\ast+1$ distinct right
$\str{b}$-runs. If it has no $\str{b}/\str{b}$ equality cell in a
genuine right $P_{\mathrm{out}}$, all these witnesses lie in right
blocks. They cannot lie in two blocks: between them the path would
traverse a heavy $\str{c}$-run of $Q_{\mathrm{out}}$ while its left
coordinate stays in $P_{\mathrm{out}}$. Every such cell mismatches,
giving cost at least $H$. One block has at most $r_\ast$ runs,
so it cannot supply all witnesses. Choose a resulting
$\str{b}/\str{b}$ equality cell in a genuine right
$P_{\mathrm{out}}$ as the leading anchor. The symmetric argument
gives a $\str{c}/\str{c}$ equality cell between the trailing left
$Q_{\mathrm{out}}$ and a genuine right $Q_{\mathrm{out}}$.

Index right guards from $0$ to $k$, and let $j,j'$ be the guard
indices of the leading and trailing anchors. The leading anchor
comes first on the left axis, so it also comes first on the right
axis. Within one right guard, however, $Q_{\mathrm{out}}$ precedes
$P_{\mathrm{out}}$. The two anchors cannot lie in the
same guard, and $j'\ge j+1$.

Suppose $j'\ge j+2$. A complete right guard lies between the
anchors and is traversed while the left coordinate remains in the
strip. Its $Q_{\mathrm{out}}$ part has $r_\ast+1$ heavy
$\str{c}$-runs, requiring that many distinct left witness runs.
The block $A$ has at most $r_\ast$ runs, and the leading left tag
contains no $\str{c}$. Hence some $\str{c}$-witness lies in the
trailing left $Q_{\mathrm{out}}$. Likewise, some $\str{b}$-witness
for the right guard's $P_{\mathrm{out}}$ lies in the leading left
$P_{\mathrm{out}}$. On the right, the $\str{c}$-cell comes before
the $\str{b}$-cell. On the left, the $\str{b}$-cell comes before
the $\str{c}$-cell. This contradicts monotonicity, so $j'=j+1$.

The subpath between the anchors now traverses the complete blocks
$A$ and $B_{j+1}$. On both axes, the leading fragment contains
no $\str{c}$ and the trailing fragment contains no $\str{b}$.
Lemma~\ref{gate:endpoint-isolation} gives $\DTW(A,B_{j+1})<H$.
\end{proof}

\subsubsection{Combining dense lists}

The doubled right list and the left padding below follow the
cyclic-alignment construction of Abboud and
Backurs~\cite{AbboudBackurs2017}. We use capped DTW encodings and
heavy guards so that false local comparisons incur additive costs
in disjoint left-coordinate strips.

Take capped encodings $A_1,\ldots,A_p$ and $B_1,\ldots,B_p$.
Let $U=\shrink(B_j)$ be their common right collapse and put
$G^{\mathrm{light}}=\shrink(G^{\mathrm{heavy}})$.
Define left blocks and right blocks by
\[
 S_i=\begin{cases}
 U,&1\le i\le p,\\
 A_{i-p},&p<i\le2p,\\
 U,&2p<i\le3p,
 \end{cases}
 \qquad T_j=B_{1+((j-1)\bmod p)}\quad(1\le j\le2p).
\]
For $0\le i\le3p$, set $G_i=G^{\mathrm{heavy}}$ when
$p\le i\le2p$ and $G_i=G^{\mathrm{light}}$ otherwise.
Set $G'_j=G^{\mathrm{heavy}}$ for $0\le j\le2p$. Define
\[
 X=G_0S_1G_1\cdots S_{3p}G_{3p},
 \qquad
 Y=G'_0T_1G'_1\cdots T_{2p}G'_{2p}.
\]
\begin{proposition}\label{pack:packing}
The strings $X,Y$ have total expanded length and construction time
$O(pHr_\ast)$, and satisfy:
\begin{enumerate}
\item If every local pair has distance at least $H$, then
      $\DTW(X,Y)\ge pH$.
\item If at most $\theta p^2$ local pairs have distance greater than
      $r_\ast$, then
      \[
       \DTW(X,Y)\le p\bigl(6r_\ast+4+\theta(H+r_\ast)\bigr).
      \]
\end{enumerate}
The strings can be padded to equal length without changing their
distance and with at most a factor-two increase in total length.
\end{proposition}

\begin{proof}
For soundness, the strip for $A_i$ consists of the $P_{\mathrm{out}}$
half of $G_{p+i-1}$, the block $A_i$, and the $Q_{\mathrm{out}}$ half
of $G_{p+i}$. Charge $A_i$ the total cost of cells whose left
coordinate lies in this strip. These cells form a contiguous segment traversing the whole
strip. Adjacent strips use different halves of their shared guard,
so their sets of cells are disjoint
(Figure~\ref{fig:packing-strips}). A charge below $H$ would give
$\DTW(A_i,T_j)<H$ for some $j$ by Lemma~\ref{pack:extraction},
contrary to the premise. Summing the $p$ charges proves the lower
bound.

For completeness, each local pair occurs in exactly one of the
$p$ cyclic matchings. At most $\theta p^2$ pairs have distance
greater than $r_\ast$, so the average number of such pairs in a
matching is at most $\theta p$. Choose a shift with at most
$\theta p$ such comparisons and represent it by the window
$T_{t+1},\ldots,T_{t+p}$, where $0\le t<p$.
The dummy blocks before and after the active list accommodate
the blocks outside this window. The alignment has five parts.
\[
\begin{array}{c|c|c}
\text{part of }X&\text{comparison on }Y&\text{number of blocks}\\
\text{unmatched prefix}&\text{hold the first }\str{a}&p-t\\
\text{matched prefix}&\text{first }t\text{ right blocks}&t\\
\text{active blocks}&\text{selected window}&p\\
\text{matched suffix}&\text{last }p-t\text{ right blocks}&p-t\\
\text{unmatched suffix}&\text{hold the last }\str{b}&t
\end{array}
\]
Traverse the unmatched prefix dummies and their preceding light
guards while holding the first $\str{a}$ of $Y$. At the next left
guard, match $G_{p-t}$ to $G'_0$. Match the remaining $t$ prefix
dummies to $T_1,\ldots,T_t$, including the intervening guards.
These comparisons cost zero because the corresponding collapses
agree.

Next compare each $A_i$ with $T_{t+i}$, matching the heavy guards
between them at zero cost. Match the first $p-t$ suffix dummies
to the remaining right blocks and guards. Finally, traverse the
last $t$ suffix dummies and their following light guards while
holding the last $\str{b}$ of $Y$. The two traversals with a fixed
right coordinate join the complete block paths at their endpoints.

Only the unmatched prefix and suffix incur cost for dummy blocks
and guards. They contain $(p-t)+t=p$ dummies and the same number
of light guards. Their cost is at most
$p(|U|+4(r_\ast+1))$.
The active comparisons cost at most
$pr_\ast+\theta p(H+r_\ast)$ by the good-pair and cap bounds.
Using $|U|\le r_\ast$ gives the stated upper bound.

There are $O(p)$ blocks and guards, each of length $O(Hr_\ast)$,
which proves the length and construction bounds. Both strings end
in $\str{b}$, so extending the shorter final run preserves their distance
by Lemma~\ref{gate:path-facts} and increases total length by at most
a factor two.
\end{proof}

\begin{figure}[htbp]
\centering
\begin{tikzpicture}[x=1.5cm,y=1cm,font=\small,
 part/.style={draw,minimum width=1.7cm,minimum height=.6cm}]
\node[part,fill=blue!8] at (0,0) {$P_{\mathrm{out}}$};
\node[part,fill=blue!8] at (1.4,0) {$A_i$};
\node[part,fill=blue!8] at (2.8,0) {$Q_{\mathrm{out}}$};
\node[part,fill=black!6] at (4.2,0) {$P_{\mathrm{out}}$};
\node[part,fill=black!6] at (5.6,0) {$A_{i+1}$};
\node[part,fill=black!6] at (7,0) {$Q_{\mathrm{out}}$};
\draw[dashed] (2.15,-.4) rectangle (4.85,.4);
\node at (3.5,.75) {shared guard $G_{p+i}$};
\draw[thick,blue!55!black] (-.57,-.55)--(-.57,-.75)--(3.37,-.75)--(3.37,-.55);
\draw[thick] (3.63,-.55)--(3.63,-.75)--(7.57,-.75)--(7.57,-.55);
\node at (1.4,-1.08) {strip for $A_i$};
\node at (5.6,-1.08) {strip for $A_{i+1}$};
\end{tikzpicture}
\caption{Adjacent strips use different halves of their shared guard,
so no path cell is charged twice. Their right-coordinate ranges may
overlap. Each displayed tag is heavy.}
\label{fig:packing-strips}
\end{figure}
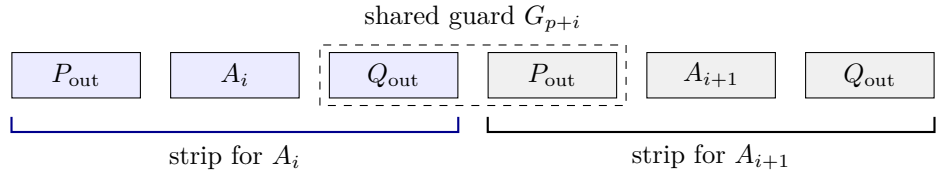

\begin{lemma}\label{lem:constant-gap}
Fix separate encodings of run bound $r$ and an integer $q\ge1$.
Given two nonempty lists of $p$ inputs, one can
construct equal-length strings
of total length $O(qpr^2)$ and a threshold $B_0$ in $O(qpr^2)$ time.
Every $q$-approximate estimate is at least $B_0$ if all input pairs
are false, and below $B_0/2$ if at most $p^2/(16q)$ pairs are false.
\end{lemma}

\begin{proof}
Use the capped encodings with $r_\ast=100r$, choose
$H=128q(r_\ast+1)$, and set $B_0=pH$.
Proposition~\ref{pack:packing} gives the all-false lower bound. In
the dense case, with $\theta\le1/(16q)$,
\[
 \frac{q\DTW(X,Y)}{pH}
 \le\frac{q(6r_\ast+4)}H
       +q\theta\left(1+\frac{r_\ast}H\right)
 <\frac6{128}+\frac1{16}\left(1+\frac1{128q}\right)<\frac12.
\]
The construction, including padding, takes
$O(pHr_\ast)=O(qpr^2)$ time.
\end{proof}

\subsection{Density amplification}\label{sec:density-amplification}

Lemma~\ref{lem:constant-gap} requires a false-pair fraction at most
$1/(16q)$. We obtain this from constant true-pair density using a
fixed family of index translations, chosen independently of the
predicate values.

We use the standard small-bias hitting bound stated by Ben-Sasson
and Viola~\cite{BenSassonViola2014}, together with the explicit
small-bias construction of Alon, Goldreich, H{\aa}stad, and
Peralta~\cite{AlonEtAl1992}.

\begin{lemma}\label{lem:small-bias}
For integers $d\ge1$ and $b\ge2$, put $\varepsilon=1/b$.
One can deterministically construct a nonempty indexed family $S$
in $\bits^d$ of size $O(d^2b^2)$ in
$O(d^3b^2\polylog(db))$ time, such that every nonempty $V\subseteq\bits^d$ of
density $\mu$ satisfies
\[
 \Pr_{z\in\bits^d}\bigl[(z\oplus S)\cap V=\varnothing\bigr]
 \le\frac{\varepsilon^2}{\mu}.
\]
Here $z$ is uniform, and $\oplus$ is coordinatewise addition modulo two.
\end{lemma}

\begin{proof}
Use the powering construction of Alon et al.~\cite{AlonEtAl1992}
to obtain an $\varepsilon$-biased family:
for every nonzero $a\in\bits^d$,
$|\mathbb E_{s\in S}(-1)^{a\cdot s}|\le\varepsilon$, where the
expectation counts indices with multiplicity. The construction indexes
the family by two elements of a field of size $O(db)$, giving
$|S|=O(d^2b^2)$. Generating all $d$ coordinates for each index,
including the field construction, takes $O(d^3b^2\polylog(db))$ time.
Let $f(z)=\ind[z\in V]$ and $g(z)=\mathbb E_{s\in S}f(z\oplus s)$.
The quantity $g(z)$ is the fraction of the translated family that
lies in $V$. The translation misses $V$ exactly when $g(z)=0$.
Also, $\mathbb E_z g(z)=\mu$.

Write $\widehat f(a)=\mathbb E_z f(z)(-1)^{a\cdot z}$ for the
Fourier coefficients of $f$, and define $\widehat g(a)$ in the
same way. Averaging over $S$ gives, for every nonzero $a$,
\[
 \widehat g(a)=\widehat f(a)\mathbb E_{s\in S}(-1)^{a\cdot s}.
\]
Parseval's identity and the bias bound give
\[
 \mathbb E_z(g(z)-\mu)^2
 \le\varepsilon^2\mathbb E_z(f(z)-\mu)^2
 =\varepsilon^2\mu(1-\mu)\le\varepsilon^2\mu.
\]
When the translation misses $V$, its squared deviation from the
mean is $\mu^2$. Hence
\[
 \Pr_z[g(z)=0]\mu^2\le\mathbb E_z(g(z)-\mu)^2.
\]
Combining this with the variance bound and dividing by $\mu^2$
proves the claim.
\end{proof}

\begin{lemma}\label{lem:amplified-gap}
Fix separate encodings of run bound $r$ and an integer $q\ge1$.
Given two nonempty lists of $p$ inputs, one can deterministically
construct equal-length strings over $\Sigma$ and a
positive integer threshold $B_0$ in time
\[
 \widetilde O(pr^2q^3).
\]
The total expanded length satisfies the same bound. Every
$q$-approximate estimate is at least $B_0$ if all input pairs are
false, and below $B_0/2$ if at least half are true.
\end{lemma}

\begin{proof}
Put $t=\max\{1,\lceil\log p\rceil\}$ and $p'=2^t$.
Pad each list to length $p'$ by repeating one of its entries, and index
the padded lists $(u_i),(v_j)$ by $\bits^t$. The all-false case is
preserved. Since $p\le p'\le2p$, at least half true pairs before
padding give a set
\[
 V=\{(i,j):F(u_i,v_j)=1\}\subseteq\bits^{2t}
\]
of density at least $1/8$ afterwards.

We need the amplified false-pair fraction to be at most $1/(16q)$.
At density at least $1/8$, Lemma~\ref{lem:small-bias} bounds this
fraction by $8\varepsilon^2$. A bias of order $q^{-1/2}$
suffices. Apply the lemma with $d=2t$ and
$\varepsilon=1/(16\lceil\sqrt q\rceil)$. The resulting family has
$k=O(qt^2)$ indices. Write its members as
$(s_\ell,s'_\ell)\in\bits^t\times\bits^t$, for $1\le\ell\le k$,
and define
\[
 F'(i,j)=\bigvee_{\ell=1}^k
               F(u_{i\oplus s_\ell},v_{j\oplus s'_\ell}).
\]
All amplified pairs are false in the all-false case. In the dense
case, $F'(i,j)=0$ exactly when $((i,j)\oplus S)\cap V$ is empty,
so the false fraction is at most
$\varepsilon^2/(1/8)\le1/(32q)<1/(16q)$.

To encode the amplified predicate, the left encoder at index $i$
uses children $X_F(u_{i\oplus s_\ell})$, and the right encoder
at index $j$ uses children $Y_F(v_{j\oplus s'_\ell})$.
Each side can construct its children from its own list, without
evaluating the cross-list predicate values. The left children
have a common collapse, as do the right children, because they
all use the original encoders. Proposition~\ref{gate:or}
and Lemma~\ref{gate:collapse-preservation} give separate encodings
with run bound
\[
 r'=(2k+1)(5r+8)+2=O(kr).
\]
In particular, true pairs cost at most $r+2(5r+8)\le r'$, and false
pairs cost at least $H$. Each encoder writes its output in $O(Hr')$ time.

Applying Lemma~\ref{lem:constant-gap} gives
the required strings and threshold in $O(qp'(r')^2)=O(qpk^2r^2)$
time and length. Substituting $k=O(qt^2)$ gives
$O(pr^2q^3t^4)$. The family construction takes
$O(qt^3\polylog(qt))$ time, which also fits within
$\widetilde O(pr^2q^3)$.
\end{proof}

\subsection{Fine-grained consequences}\label{sec:hardness}

\begin{corollary}\label{lem:fine-grained-transfer}
For every fixed $\kappa\in(0,1/6)$ and $\delta\in(0,1)$, there are
constants $a,\varepsilon>0$ with the following
property. Suppose DTW over $\Sigma$ admits a deterministic
$N^{\kappa\delta}$-approximation in $O(N^{2-\delta})$ time on
two strings of length $N$.
Consider any two lists of $p$ inputs to a predicate with separate
encodings of run bound $r\le p^a$.
One can then deterministically distinguish no true pair from at least half true pairs in
$O(p^{2-\varepsilon})$ time.
\end{corollary}

\begin{proof}
Set $q=\lceil p^\gamma\rceil$. Lemma~\ref{lem:amplified-gap} bounds
both the common output length $N$ and the construction time by
$p^{1+\beta+o(1)}$, where $\beta=3\gamma+2a$.
There are two requirements on $\gamma,a$. The approximation
factor must be smaller than $q$, and the DTW running time must
remain strictly subquadratic in $p$. These require, respectively,
\[
 \kappa\delta(1+\beta)<\gamma,
 \qquad (2-\delta)(1+\beta)<2.
\]
With $a=0$, both inequalities hold precisely when
\[
 \frac{\kappa\delta}{1-3\kappa\delta}
 <\gamma<\frac{\delta}{3(2-\delta)}.
\]
This interval is nonempty because $6\kappa<1$. Choose a positive
rational $\gamma$ in the interval. Since both inequalities are
strict, they remain valid for a sufficiently small constant $a>0$.

Use the reduction with these choices. The first inequality gives
$N^{\kappa\delta}<q$ for sufficiently large $p$, so the
approximation separates the two cases at threshold $B_0/2$.
Its running time is at most
\[
 p^{(1+\beta)(2-\delta)+o(1)}.
\]
The second inequality makes this exponent strictly below two.
Since $2-\delta>1$, the construction exponent $1+\beta$ is also
strictly below two. The fixed margins absorb the $o(1)$ terms,
giving the claimed saving $\varepsilon>0$.
\end{proof}

\begin{proof}[Proof of Theorem~\ref{thm:explicit}]
Suppose the stated approximation algorithm exists, and let
$a,\varepsilon$ be given by Corollary~\ref{lem:fine-grained-transfer}.
Choose a fixed $0<\alpha<a$. For any fixed $c_0>0$, consider GBD
with $K=\lceil p^\alpha\rceil$ and $L=\lceil c_0\log p\rceil$.
Proposition~\ref{prop:predicate-encoders} gives run bound
$r=O(KL)\le p^a$ for sufficiently large $p$. The dense GBD case
has at least half true pairs, so the corollary solves GBD in
$O(p^{2-\varepsilon})$ time. The constants $\alpha,\varepsilon$ do not
depend on $c_0$, contradicting Hypothesis~\ref{hyp:gbd}.
\end{proof}

\subsection{Circuit consequences}\label{sec:formula-encodings}

The OR construction requires common child collapses on each side.
For arbitrary formulas, we choose collapses that depend
only on the depth bound and the side. A fixed gate layout gives
this uniformity even when the formulas use different root gates.

\begin{lemma}\label{lem:formula-encoders}
Let $\Phi(x,y)$ be a Boolean formula of depth at most $d$, with
two-input AND and OR gates and literals or constants at the leaves.
The variables are partitioned between $x$ and $y$.
There are separate encoders for $\Phi$ satisfying
Definition~\ref{def:encoders} with a run bound $\tau_d=2^{O(d)}$.
Their side collapses depend only on $d$, independently of the formula,
the variable partition, and the assignments.
\end{lemma}
\begin{proof}
We induct on the depth bound $d$. At depth zero, use the two-slot literal encoders from
Section~\ref{sec:literal-encoders}. Treat a constant as a literal
assigned to the first side. These encoders have a constant run bound
$\tau_0$ and the required fixed side collapses.

At every positive depth bound, use an OR of two AND gadgets.
For a root with children $f,g$, write
\[
 f\wedge g=(f\wedge g)\vee(0\wedge0),\qquad
 f\vee g=(f\wedge1)\vee(g\wedge1).
\]
For a leaf $f$, use $f=(f\wedge1)\vee(0\wedge0)$.
The leaf identity fills the remaining
levels when the formula is shorter than the depth bound.
Encode all four inputs, including constants, with depth bound
$d-1$. By induction, their collapses agree on each side. Apply both AND
gadgets with the same child run bound $\tau_{d-1}$, and use their
common output run bound for the final OR gadget. The two AND outputs
have common side collapses, as the OR gadget requires.
Lemma~\ref{gate:collapse-preservation} shows that the final side
collapses are distinct. Since the layout and run parameters depend
only on $d$, so do those collapses.

If $\Phi$ is true, one AND has two true children. Its distance is
$O(\tau_{d-1})$, and selecting it in the OR adds
$O(\tau_{d-1})$. If $\Phi$ is false, both ANDs have a false child,
so AND and OR soundness give distance at least $H$.
The run bound after the two stages is also $O(\tau_{d-1})$, and the
endpoints and allowed run lengths are preserved.
Choose an absolute integer $C>4$ large enough to bound both the
run count and the true-pair distance, and set
$\tau_d=C\tau_{d-1}$. This proves the encoding guarantees and
$\tau_d=2^{O(d)}$.

Let $T_d$ bound the construction time of one encoder. At depth
$d>0$, it makes four recursive calls and writes $O(H\tau_d)$
symbols, so
\[
 T_d\le4T_{d-1}+O(H\tau_d).
\]
Since $\tau_d=C\tau_{d-1}$ with $C>4$, summing the recursive
costs gives $T_d=O(H\tau_d)$. The fixed collapse strings can be
generated along the same recursion. Each encoder uses only its
own assignment.
\end{proof}

\begin{corollary}\label{cor:circuits}
Fix $\kappa\in(0,1/6)$.
Suppose that, for some fixed $\delta\in(0,1)$, DTW on two strings of
length $N$ over the uniform metric on $\Sigma$ admits a deterministic
$N^{\kappa\delta}$-approximation in $O(N^{2-\delta})$ time.
Fix $k\ge1$ and a polynomial-time computable integer function
$h(v)\ge1$ with $h(v)=o(v)$.
Then $\mathsf{E}^{\mathsf{NP}}$ is not contained in the class of functions computed by
nonuniform fan-in-two Boolean circuits of size $O(v^k)$ and depth
$O(h(v))$.
Also, $\mathsf{E}^{\mathsf{NP}}$ does not have nonuniform linear-size Valiant
series-parallel circuits.
\end{corollary}
\begin{proof}
Let $a,\varepsilon>0$ be given by Corollary~\ref{lem:fine-grained-transfer}
for the assumed DTW algorithm. Consider a polynomial-size,
fan-in-two circuit on $v$ variables whose depth is bounded by a
polynomial-time computable function $d(v)=o(v)$. Unfold it into a
formula and push negations to the leaves. The formula has depth
at most $d(v)$ and size $2^{O(d(v))}$.
Add one unused variable if $v$ is odd. Split the variables into two
equal groups and list all partial assignments.
Each list has size $p=2^{\lceil v/2\rceil}$.
A uniformly chosen pair induces the uniform distribution on assignments.
Lemma~\ref{lem:formula-encoders} gives run bound
$r=2^{O(d(v))}=p^{o(1)}\le p^a$ for sufficiently large $v$.
Corollary~\ref{lem:fine-grained-transfer} distinguishes an unsatisfiable
circuit from one satisfied by at least half of its assignments in
$O(p^{2-\varepsilon})=O(2^{(1-\varepsilon/2)v})$ time. Reading the circuit and
unfolding the formula take $2^{o(v)}$ time and fit within this bound.

We apply the criterion of Ben-Sasson and
Viola~\cite{BenSassonViola2014} to a family $\mathcal C_v$ of
functions with polynomial-size descriptions. For this application,
we must verify closure under input projections and complementation,
and give an algorithm for the following promise problem. The input
is an AND of polynomially many three-input ORs of functions from
$\mathcal C_{v+O(\log v)}$, with inputs projected to $v$ variables.
The algorithm must distinguish a circuit accepting every assignment
from one accepting at most a $v^{-10}$ fraction of assignments, in
$2^v/v^{\omega(1)}$ time. The criterion then implies that
$\mathsf{E}^{\mathsf{NP}}$ contains a function outside the family.

To absorb constant factors in the target size and depth bounds,
choose an integer $k'>k$ and put
\[
 h'(v)=\left\lceil\sqrt{v(h(v)+1)}\right\rceil.
\]
Let $\mathcal C_v$ consist of the functions represented by fan-in-two
AND/OR circuits of size at most $v^{k'}$ and depth at most $h'(v)$,
with literals and constants as inputs. These functions have
polynomial-size descriptions. Replacing an input variable by a
constant or a possibly negated input variable changes only the input
labels, preserving size and depth. To complement a circuit, interchange
AND and OR and complement every input literal or constant.
De Morgan's laws give the complemented function with the same size
and depth. Both operations are polynomial-time computable on the
descriptions, so the family has the closure properties required by
the criterion. Also, $h'(v)=o(v)$ and $h(v)=o(h'(v))$.

Computing both polarities of each gate removes internal NOT gates
with at most twice the size and no increase in depth.
The choices of $k'$ and $h'$ ensure that every circuit family of
size $O(v^k)$ and depth $O(h(v))$ computes functions in $\mathcal C_v$
for sufficiently large $v$.

It remains to give the algorithm required by the criterion.
Balance the outer AND and complement the resulting circuit.
The two cases become unsatisfiability and acceptance on at least
a $1-v^{-10}$ fraction of assignments. The complemented circuit
has polynomial size and a polynomial-time computable depth bound
$o(v)$. To bound the depth, take the maximum of $h'(w)$ over
indices within $O(\log v)$ of $v$, then add $O(\log v)$ for the
outer gates. The gap-circuit algorithm from the first paragraph
distinguishes the two cases in
$O(2^{(1-\varepsilon/2)v})$ time, which is within
$2^v/v^{\omega(1)}$. This verifies the algorithmic requirement,
so the criterion gives a function in $\mathsf{E}^{\mathsf{NP}}$
outside $\mathcal C_v$.

This proves the stated lower bound.
For the series-parallel conclusion,
the assumed DTW algorithm refutes GBD by Theorem~\ref{thm:explicit},
and the connection of Abboud and Backurs~\cite{AbboudBackurs2017}
gives the claimed circuit lower bound.
\end{proof}

\section{Hardness for run-length encoded strings}\label{sec:rle}

\subsection{From membership to DTW}\label{sec:compiler}

The membership instances use the alphabet
\[
 \Gamma=\{\str{Z}_{11},\str{Z}_{12},\str{Z}_{21},\str{Z}_{22},\str{U},\str{B}_\star\}.
\]
For $x,y\in\{1,2\}$, the symbol $\str{Z}_{xy}$ encodes the
right-vector coordinate pair $(2-x,2-y)$.
For $(\alpha,\beta)\in\bits^2$, the set $C_{\alpha\beta}$ contains
exactly the $\str{Z}$-symbols whose coordinate pair is orthogonal
to $(\alpha,\beta)$, together with $\str{U}$ and, when $\alpha=0$,
the marker $\str{B}_\star$. Thus
\[
\begin{array}{ll}
 C_{00}=\Gamma,&C_{01}=\{\str{Z}_{12},\str{Z}_{22},\str{U},\str{B}_\star\},\\[2pt]
 C_{10}=\{\str{Z}_{21},\str{Z}_{22},\str{U}\},&C_{11}=\{\str{Z}_{22},\str{U}\}.
\end{array}
\]
All four sets contain the filler $\str{U}$. The marker
$\str{B}_\star$ fixes group boundaries. Its presence in $C_{01}$
does not affect source soundness, which uses $\str{B}_\star\notin C_{11}$.
The required types are $\Gamma,C_{01},C_{10},C_{11},\{\str{U}\}$,
and $\{\str{B}_\star\}$; the optional types are
$(\epsilon\mid \str{U})$ and $(\epsilon\mid\Gamma)$.
Each empty factor will use a separate dummy interval in the DTW
construction. A bound on consecutive empty factors controls
the number of dummy intervals that we need between text symbols.

\begin{proposition}\label{prop:three-symbol-compiler}
Let $T\in\Gamma^{L_T}$ and $P=D_1\cdots D_s$ use the factor types
above, with $L_T,s\geq1$. Fix an integer $\ell\geq1$.
Suppose that, whenever $T\in L(P)$, there is an accepting derivation
with at most $L_T$ empty factors at either end and at most $\ell$
between consecutive text symbols. For every integer $H\geq2$,
one can construct two run-length encoded strings $X,Y$ over $\Sigma$ in
$O(n+m)$ time. Their total run count satisfies $n+m=O(s+L_T\ell)$,
and their distance satisfies
\[
 T\in L(P)\ \Longrightarrow\ \DTW(X,Y)\leq n+m-1,
 \qquad
 T\notin L(P)\ \Longrightarrow\ \DTW(X,Y)\geq H.
\]
The run labels and $n+m$ are independent of $H$. Every run has length
one or $H$, and the total expanded length is $N+M=\Theta((n+m)H)$.
\end{proposition}

\paragraph{Symbol and factor gadgets.}
For $v\in\Gamma$, the symbol gadget $T_v$ replaces each position
of $\operatorname{code}(v)$ in Table~\ref{tab:real} by a run of
length $H$. A dummy gadget is the light symbol $\str{a}$.
Let $A_j$ be the alternating string of odd length $j$ that begins
and ends with $\str{a}$. A factor profile $(w,J)$ consists of an
alternating string $w$ over $\{\str{a},\str{b}\}$ and a set
$J\subseteq[|w|]$ of heavy positions. Position $i$ becomes $w_i^H$
if $i\in J$ and remains light otherwise.
The profiles in Table~\ref{tab:factors} each have a heavy run.

A profile $(w,J)$ \emph{accepts} a code if deleting some
positions of $w$ outside $J$, then collapsing equal neighbors,
produces that code.

\begin{table}[H]
\centering
\begin{minipage}[t]{.39\textwidth}
\centering
\vspace{0pt}
\small
\begin{tabular}{ccl}
\toprule
Symbol & $\operatorname{code}(v)$ & Length\\
\midrule
$\str{U}$ & $\str{aba}$ & $3$\\
$\str{Z}_{22}$ & $\str{ababa}$ & $5$\\
$\str{Z}_{21}$ & $\str{abababa}$ & $7$\\
$\str{Z}_{11}$ & $\str{ababababa}$ & $9$\\
$\str{Z}_{12}$ & $\str{ab}$ & $2$\\
$\str{B}_\star$ & $\str{ba}$ & $2$\\
\bottomrule
\end{tabular}
\caption{Symbol codes. Every run in a symbol gadget has length $H$.}
\label{tab:real}
\end{minipage}\hfill
\begin{minipage}[t]{.59\textwidth}
\centering
\vspace{0pt}
\small
\begin{tabular}{lcc}
\toprule
Factor & Word $w$ & Heavy positions $J$\\
\midrule
Required $\Gamma$ & $A_9$ & $\{2\}$\\
Required $C_{01}$ & $A_5$ & $\{2\}$\\
Required $C_{10}$ & $A_7$ & $\{1,2,7\}$\\
Required $C_{11}$ & $A_5$ & $\{1,2,5\}$\\
Required $\{\str{U}\}$ & $\str{aba}$ & $\{1,2,3\}$\\
Required $\{\str{B}_\star\}$ & $\str{ba}$ & $\{1,2\}$\\
Optional $\str{U}$ & $\str{aba}$ & $\{1,3\}$\\
Optional $\Gamma$ & $A_9$ & $\{5\}$\\
\bottomrule
\end{tabular}
\caption{Factor profiles. Unlisted positions have length one.}
\label{tab:factors}
\end{minipage}
\end{table}

\begin{lemma}\label{lem:table}
For each $v\in\Gamma$, a profile in Table~\ref{tab:factors} accepts
$\operatorname{code}(v)$ if and only if its factor can emit $v$.
The optional profiles, and only those profiles, can give every
heavy run an equality cell against the dummy.
\end{lemma}
\begin{proof}
For required $\Gamma$, retain the prefix of length $3,5,7$, or $9$
for the four odd-length codes. For $\str{ab}$, retain positions $1,2$.
For $\str{ba}$, retain positions $2,3$.
Each choice retains the heavy position $2$.
Required $C_{01}$ has the same witnesses up to length five and
cannot produce a longer code, so it accepts precisely the codes
for $\str{U},\str{Z}_{22},\str{Z}_{12},\str{B}_\star$.

The required $C_{10}$ profile illustrates how the heavy positions
restrict the accepted codes. Its word is $A_7$, with heavy positions
$1,2,7$. Retaining the two endpoint $\str{a}$-positions forces
the code to start and end with $\str{a}$, and retaining position
$2$ forces it to contain $\str{b}$. Its length is at most seven.
Among the symbol codes in Table~\ref{tab:real}, these conditions
leave precisely $\str{U},\str{Z}_{22},\str{Z}_{21}$.
For example, retaining $1,2,7$ gives $\str{aba}$, the code for
$\str{U}$. Retaining $1,2,3,4,7$ gives $\str{ababa}$, the code
for $\str{Z}_{22}$, and retaining all seven positions gives the
code for $\str{Z}_{21}$. Required
$C_{11}$ is the length-five version, retaining $1,2,5$ for $\str{U}$ or
all five positions for $\str{Z}_{22}$. The two singleton required profiles
have no light positions, so their run collapses cannot change.

Optional $\str{U}$ retains its two endpoint $\str{a}$-positions. Keeping its $\str{b}$ gives
$\str{aba}$; deleting it gives $\str{a}$. For optional $\Gamma$, retain
$3,4,5$ for $\str{U}$. For the other odd-length codes, retain the
prefix of length $5,7$, or $9$. For $\str{ab}$, retain positions $5,6$.
For $\str{ba}$, retain positions $4,5$.
Every choice retains the heavy position $5$.

To align an optional profile with the dummy, hold the dummy $\str{a}$
while consuming the profile. Every heavy run has label $\str{a}$ and receives
an equality cell; only light $\str{b}$ positions incur cost. Every required
profile has a heavy $\str{b}$-run, which has no equality cell against a dummy.
\end{proof}

\paragraph{The output strings.}
Write $T=v_1\cdots v_{L_T}$ and use $T_{v_1},\ldots,T_{v_{L_T}}$ as the
real intervals. Place $L_T$ dummy intervals before the first real
interval, $L_T$ after the last, and $\ell$ in each internal gap.
There are $D_{\mathrm{dum}}=2L_T+(L_T-1)\ell$ dummies.
An accepting derivation satisfying
the proposition's hypothesis assigns each emitting factor to its
text symbol and each empty factor to a distinct dummy in the same
gap. This assignment is injective, preserves order, and covers every
real interval.

Let $F_1,\ldots,F_s$ be the factor gadgets and
$I_1,\ldots,I_{L_T+D_{\mathrm{dum}}}$ the resulting list of real and dummy gadgets.
Set
\begin{align*}
X&=\str{c}^H F_1\str{c}^H F_2\cdots \str{c}^H F_s\str{c}^H,\\
Y&=\str{c}^H I_1\str{c}^H I_2\cdots \str{c}^H I_{L_T+D_{\mathrm{dum}}}\str{c}^H.
\end{align*}
No gadget contains $\str{c}$, so no runs merge at these boundaries.
Figure~\ref{fig:rle-overview} shows the assignment of factors to intervals.

\paragraph{Soundness.}
Fix a warping path of cost less than $H$. A factor and an interval
are \emph{incident} if an equality cell joins their positions.

\begin{lemma}\label{lem:matching}
Incidence has the following properties.
\begin{enumerate}
\item It is an ordered matching covering every factor and every real
interval. Every heavy factor run has an equality cell in its
factor's unique matched interval.
\item If a factor with profile $(w,J)$ is matched to $T_v$, then
$(w,J)$ accepts $\operatorname{code}(v)$.
\end{enumerate}
\end{lemma}
\begin{proof}
Suppose a factor has equality cells in two distinct
intervals. Between those cells, the right coordinate traverses
a full separator $\str{c}^H$, while the left coordinate stays
in the factor. Since the factor contains only $\str{a},\str{b}$,
all cells in that traversal mismatch, giving cost at least $H$.
Interchanging the axes shows that an interval cannot be incident
to two factors. Incidence is a matching. Monotonicity orders its pairs.

Every factor and every real interval has a heavy run of
$\str{a}$ or $\str{b}$. By Lemma~\ref{gate:heavy-witnesses},
that run has an equality cell, whose other coordinate cannot lie
in a separator. The matching covers every factor and
every real interval. Applying the same lemma to each heavy factor
run shows that it has an equality cell in the factor's unique
matched interval.

For the second assertion, fix a profile $(w,J)$ matched to $T_v$.
We recover a retained subsequence of $w$ that collapses to the code.
Every heavy factor run has an equality cell in $T_v$, and every
run of $T_v$ has an equality cell in the factor. A factor run
cannot have equality cells in two runs of $T_v$. Such code runs
would have the same label, with an opposite-label run between
them. Traversing that intervening heavy run while staying in the
factor run would cost at least $H$.

Retain exactly those positions of $w$ whose runs have equality
cells in $T_v$, and assign each to the unique code run containing
those cells. All heavy positions $J$ are retained. Monotonicity
makes the assigned code-run indices nondecreasing.
Every code run has an equality witness, so every code-run index occurs.
The retained positions assigned to one code run form a consecutive
group with that run's label. These groups visit all code runs in
order, so collapsing them gives $\operatorname{code}(v)$.
\end{proof}

By Lemmas~\ref{lem:matching} and~\ref{lem:table}, each factor
matched to a real interval can emit its symbol, and each factor
matched to a dummy is optional. Let the latter factors emit
$\epsilon$. The ordered matching covers every factor and every
text symbol, so these choices form an accepting derivation.
This proves $T\notin L(P)\Rightarrow\DTW(X,Y)\geq H$.

\paragraph{Completeness.}
Suppose $T\in L(P)$, and choose the accepting derivation and
ordered assignment described above. We first work on the run
collapses. We need a path that visits all positions and gives
every position to be expanded into a heavy run an equality cell.
The expansion assertion of Lemma~\ref{gate:heavy-witnesses} will
then preserve its cost.

For a factor emitting $v$, Lemma~\ref{lem:table} supplies retained
profile positions whose subsequence collapses to
$\operatorname{code}(v)$. The path must also visit the deleted
positions. These positions are light. The maximal groups of equal symbols in the
retained subsequence correspond to the code positions.
Scan the entire profile from left to right. Match each retained
position to its group's code position. Consume a deleted position
while holding the most recent code position, or the first code
position if no retained position has been reached. At the first
retained position of each later group, advance both coordinates.

This gives a full path across the factor-code rectangle. Every
heavy factor position is retained and has an equality cell.
Every code position is matched to at least one retained factor position.
For an empty factor, use the optional-to-dummy path from
Lemma~\ref{lem:table}, which also gives every heavy factor position
an equality cell.

To connect consecutive assigned intervals, hold the intervening
$\str{c}$ of $\shrink(X)$ while traversing the portion of
$\shrink(Y)$ between them. This portion contains only separators
and unused dummies, since every real interval is assigned. All
its separators match the held $\str{c}$ at zero cost; the dummy
positions may cost one each. The initial and final $\str{c}$'s
handle the unused end dummies in the same way. Diagonal steps
connect these portions to the factor-interval rectangles. Every
separator on both axes now has an equality cell.

The resulting full path on the run collapses has at most
$n+m-1$ cells, each of cost at most one. Every position whose run
is to be heavy has an equality cell. Expanding those positions to
length $H$ adds no cost, giving
$\DTW(X,Y)\leq n+m-1$.

\paragraph{Size.}\label{sec:rle-run-counts}
Every gadget has a constant number of runs, and separators prevent
run mergers. The string $X$ has $n=\Theta(s)$ runs, and $Y$ has
$m=\Theta(L_T+D_{\mathrm{dum}})$ runs. Hence
$n+m=\Theta(s+L_T+D_{\mathrm{dum}})=O(s+L_T\ell)$.
Every factor has a heavy run, every run of a real interval is heavy, and each
dummy has a heavy separator. A constant fraction of all runs is
heavy, and every run has length at most $H$.
Hence $N+M=\Theta((n+m)H)$.
The run lists are written directly in $O(n+m)$ time.

\subsection{A hard membership source}\label{sec:serial-source}

We adapt Backurs and Indyk's reduction~\cite{BackursIndyk2016}
to the factor types above and keep explicit bounds on consecutive
empty factors.

\begin{proposition}\label{prop:serial-source}
Let $A=(a^1,\ldots,a^p)$ and $B=(b^1,\ldots,b^p)$ be two lists
of vectors in $\{0,1\}^d$, where $p\geq1$ and $d\geq2$ is even.
One can construct a text $T\in\Gamma^{L_T}$ and a pattern $P$
with $s$ factors of the eight types above in $O(pd)$ time.
They satisfy
\[
 T\in L(P)\quad\Longleftrightarrow\quad
 \exists i,j\in[p]\quad a^i\cdot b^j=0.
\]
We have $L_T+s=O(pd)$. Whenever $T\in L(P)$, it has an accepting
derivation with at most $L_T$ empty factors at either end and at most
$\ell=d/2+1$ between consecutive text symbols.
\end{proposition}

\begin{proof}\label{sec:source-proof}
Append a copy of the final vector to each list, obtaining size
$\bar p=p+1$. This preserves the existence of an orthogonal pair
and leaves a pattern group after every original left index.
Put $h=d/2$ and define the text blocks
\[
 E(b)=\str{U}\str{Z}_{\,2-b_1,\,2-b_2}\cdots
      \str{Z}_{\,2-b_{d-1},\,2-b_d},
 \qquad E_0=\str{U}^{h+1},\qquad E_1=\str{U}\str{B}_\star^h.
\]
The block $E(b)$ encodes the coordinates of $b$.
We call $E(0^d)$ and $E_0$ \emph{neutral blocks}. The all-$\str{U}$
block $E_0$ gives an emission available to every vector gadget.
In the text, the marker $\str{B}_\star$ occurs only in $E_1$ and
will fix the endpoints of the matched pattern core.
Set $b^j=1^d$ for $j\notin[\bar p]$, and let
\[
 T=\bigodot_{j=-\bar p}^{2\bar p}
       \bigl(E(0^d)\,E_0\,E(b^j)\,E(0^d)\,E_0\,E_1\bigr).
\]
The six blocks indexed by $j$ form text group $j$. Every block
has length $h+1$, so $L_T=|T|=O(pd)$.

The pattern gadgets are
\[
 V(a)=\str{U}C_{a_1,a_2}\cdots C_{a_{d-1},a_d},\qquad
 V_0=\str{U}(\epsilon\mid \str{U})^{h+1}\Gamma^h,\qquad
 V_1=\str{U}\str{B}_\star^h.
\]
The sets $C_{\alpha\beta}$ give
\[
 E(b)\in L(V(a))\quad\Longleftrightarrow\quad a\cdot b=0.
\]
In particular, every $V(a)$ can emit $E(0^d)$. It can also emit
$E_0$, since all four sets contain $\str{U}$.
The boundary gadget $V_1$ emits exactly $E_1$, whereas
$V(1^d)$ cannot emit $E_1$ because $\str{B}_\star\notin C_{11}$.
The gadget $V_0$ emits $t$ copies of $\str{U}$ followed by $h$ arbitrary
symbols of $\Gamma$, where $1\leq t\leq h+2$.
The choices $t=1$ and $t=h+2$ emit one and two text blocks, respectively.
All three types of pattern gadget begin with $\str{U}$.
Define the core
\[
 Q=V_1\left(\bigodot_{i=1}^{\bar p}
       [V(1^d)\,V_0\,V(a^i)\,V_0]\right)
       V(1^d)\,V(1^d)\,V_1,
\]
where the four bracketed gadgets form pattern group $i$.
We first prove that a substring of $T$ belongs to $L(Q)$ exactly
when the input has an orthogonal pair.

For soundness, suppose a substring of $T$ belongs to $L(Q)$ and fix
an accepting derivation. Since $\str{B}_\star$ occurs only in the final
$h$ positions of an $E_1$ block and $h\geq1$, the initial $V_1$
occupies an entire $E_1$ block, say in text group $\tau$.
After the initial marker, the first pattern group starts at the
first block of text group $\tau+1$. As long as every $V_0$ chooses
$t=h+2$, each pattern group consumes a complete text group.
Specifically, pattern group $i$ starts at the first block of text
group $\tau+i$ and ends at its sixth block. The table shows this
default alignment.
\[
 \begin{array}{c|cccc}
 \text{gadget}&V(1^d)&V_0&V(a^i)&V_0\\
 \text{text blocks}&1&2,3&4&5,6\\
 \text{emission}&E(0^d)&E_0E(b^{\tau+i})&E(0^d)&E_0E_1
 \end{array}
\]
If all groups followed this alignment, the final three gadgets
would have to emit $E(0^d)E_0E(b^{\tau+\bar p+1})$, the first
three blocks of text group $\tau+\bar p+1$. This is impossible
because the final $V_1$ ends in $\str{B}_\star^h$, whereas every $E(b)$
contains no $\str{B}_\star$.

At least one $V_0$ must choose $t<h+2$.
Consider the first such gadget, and let it occur in pattern group $i$.
All earlier gadgets follow
the default alignment, so this $V_0$ starts at one of the two
$E_0$ blocks in text group $\tau+i$. The next text block is
$E(b^{\tau+i})$ for the first $V_0$, or $E_1$ for the second.
Both start with $\str{U}$ followed by $h$ non-$\str{U}$ symbols.

If $2\leq t\leq h+1$, the emitted length $t+h$ consumes the
whole $E_0$ block and the first $t-1$ symbols of the next block.
The following gadget would then start at position $t$ of that
block, whose symbol is not $\str{U}$. Every pattern gadget must
start with $\str{U}$, so these choices are impossible. Hence
$t=1$, and this $V_0$ consumes exactly $E_0$.

The first such choice cannot occur in the second $V_0$ of a
pattern group. It would place the next gadget, $V(1^d)$, on
$E_1$, which that gadget cannot emit.
The first shorter choice occurs in the first $V_0$ of the pattern group.
The next gadget $V(a^i)$ now occupies the whole
block $E(b^{\tau+i})$, giving $a^i\cdot b^{\tau+i}=0$.
Duplicated vectors correspond to original entries. If $\tau+i$
is out of range, then $b^{\tau+i}=1^d$, forcing $a^i=0^d$,
which is orthogonal to every original right vector.

For completeness, choose original indices $i_*,j_*\in[p]$ with
$a^{i_*}\cdot b^{j_*}=0$.
Start the initial $V_1$ at the $E_1$ block of text group $j_*-i_*$.
Use the default alignment for pattern groups $i<i_*$. These groups
end at the sixth block of their text groups, so pattern group $i_*$
starts at the first block of text group $j_*$.
We use the orthogonal pair to change this alignment in groups
$i_*$ and $i_*+1$.
\[
 \begin{array}{c|cccc}
 \text{group}&V(1^d)&V_0&V(a^i)&V_0\\
 i_*&E(0^d)&E_0&E(b^{j_*})&E(0^d)\\
 i_*+1&E_0&E_1&E(0^d)&E_0E(b^{j_*+1})
 \end{array}
\]
Orthogonality permits $V(a^{i_*})$ to emit $E(b^{j_*})$. Every one-block
$V_0$ emission uses $t=1$, and every two-block emission uses
$t=h+2$.
Group $i_*$ consumes the first four blocks of text group $j_*$.
Group $i_*+1$ consumes its last two blocks and the first three
blocks of group $j_*+1$. The next pattern group starts
at the fourth block of a text group.

For each remaining pattern group $i=i_*+2,\ldots,\bar p$, put
$j=i+j_*-i_*$ and use the following emissions.
\[
 \begin{array}{c|cccc}
 \text{gadget}&V(1^d)&V_0&V(a^i)&V_0\\
 \text{emission}&E(0^d)&E_0E_1&E(0^d)&E_0E(b^j)
 \end{array}
\]
This group consumes blocks four through six of text group $j-1$
and blocks one through three of text group $j$. Each
remaining group starts at a fourth block and ends at a third
block of the next text group. After the last group, the final
three gadgets emit the remaining $E(0^d)E_0E_1$ blocks,
completing text group $\bar p+j_*-i_*$.
All groups used lie between $j_*-i_*$ and $\bar p+j_*-i_*$,
hence within $[-\bar p,2\bar p]$.

Finally, put
\[
 P=(\epsilon\mid\Gamma)^{L_T} Q(\epsilon\mid\Gamma)^{L_T}.
\]
Each padding block emits any word of length at most $L_T$, so
$T\in L(P)$ exactly when a substring of $T$ belongs to $L(Q)$.
There are $O(p)$ gadgets of $O(d)$ factors in $Q$, giving
$s=O(pd)$ and the claimed construction bound.

The optional factors in $Q$ occur in blocks of $h+1$ inside
$V_0$, each preceded and followed by required factors. The core
also begins and ends with required factors. In an accepting
derivation, reorder the choices within each padding block,
preserving the order of the emitting choices. Place all empty
choices first in the prefix padding and last in the suffix padding.
The factors within each padding block are identical, so this
preserves the emitted word. Every internal block
of empty factors now lies in a single optional core block, while
each end block has length at most $L_T$. The derivation has
the required bound $\ell=h+1$.
\end{proof}

\subsection{The RLE lower bound}\label{sec:randomized-calibration}

\begin{proof}[Proof of Theorem~\ref{thm:rle}(\ref{thm:rle-constant})]
Propositions~\ref{prop:serial-source} and~\ref{prop:three-symbol-compiler}
reduce Orthogonal Vectors on two lists of $p$ vectors of dimension
$d$ to DTW. The source construction has $L_T,s=\Theta(pd)$ and
$\ell=\Theta(d)$. The factor string has $n=\Theta(s)=\Theta(pd)$
runs. The text string has an extra factor of $d$ because each
internal text gap receives $\ell$ dummy intervals. Indeed,
\[
 D_{\mathrm{dum}}=2L_T+(L_T-1)\ell=\Theta(pd^2),\qquad
 m=\Theta(L_T+D_{\mathrm{dum}})=\Theta(pd^2).
\]
Hence $nm=\Theta(p^2d^3)$ and
$n+m=O(pd^2)$. Choose
$H=(\lceil C\rceil+1)(n+m)$. A successful $C$-approximate estimate
is at most $C(n+m-1)<H$ in the \textsc{Yes} case and at least $H$ in the \textsc{No}
case. Comparing the estimate with $H$ decides the source instance.
The output has $N+M=\Theta_C((n+m)^2)$.

Suppose the stated algorithm exists; it suffices to consider
$0<\delta<1$. Choose the constant $c$ supplied by
Hypothesis~\ref{hyp:rovh} for exponent saving $\delta$, and start
from an OV instance with dimension $\lceil c\log p\rceil$.
Pad vectors with zeros if necessary to make the dimension at least
two and even. The resulting dimension satisfies $d=O(\log p)$, so
\[
 nm=O(p^2\log^3 p),\qquad \log(N+M)=O(\log p).
\]
Including construction, the resulting OV algorithm runs in
\[
 O(pd^2+(nm)^{1-\delta}\polylog(N+M))
 =O(p\log^2p+p^{2-2\delta}\polylog p)
 =O(p^{2-\delta}).
\]
The resulting OV algorithm contradicts Hypothesis~\ref{hyp:rovh}
for the chosen $c$.
\end{proof}

\begin{proof}[Proof of Theorem~\ref{thm:rle}(\ref{thm:rle-large})]
Set $e=\lceil2/\eta\rceil$ and $H=(n+m)^e$ in the same reduction.
We have $N+M\leq(n+m)H=(n+m)^{e+1}$, where
$e+1=O(1/\eta)$. Since $(e+1)\eta>2$, a successful estimate on a
\textsc{Yes} instance satisfies
\[
 \widehat D\leq(n+m-1)(N+M)^{1-\eta}
 <(n+m)^{1+(e+1)(1-\eta)}<(n+m)^e=H.
\]
On a \textsc{No} instance it is at least $H$. We still have
$\log(N+M)=O_\eta(\log(n+m))$, so the running-time analysis in the proof
of part~(\ref{thm:rle-constant}) again contradicts
Hypothesis~\ref{hyp:rovh}.
\end{proof}

For any fixed $\lambda>0$, choose $\eta=2/(\lambda+2)$. Then
$H\ge(n+m)^{\lambda+2}$, so
\[
 (N+M)^{1-\eta}\ge H^{1-\eta}\ge(n+m)^\lambda.
\]
On these instances, an $(n+m)^\lambda$-approximation is also an
$(N+M)^{1-\eta}$-approximation.
Part~(\ref{thm:rle-large}) rules out such an algorithm within the
same running time.

\section*{AI assistance}
OpenAI's Codex tools assisted with developing and checking the constructions
and proofs in Sections~\ref{sec:explicit} and~\ref{sec:rle},
including the formula encodings. The tools also assisted with locating
and comparing references and with preparing and revising the text and
figures. This assistance included proposed proof steps, parameter calculations,
and separate audits of the arguments and exposition.
Anthropic's Claude also assisted with reviewing the proofs, proposing
parameter refinements, comparing related work, and refining the exposition.

\begingroup
\footnotesize
\interlinepenalty=10000
\makeatletter
\renewcommand{\@openbib@code}{%
  \setlength{\itemsep}{1pt}\setlength{\parsep}{0pt}}
\makeatother
\bibliographystyle{alphaurl}
\bibliography{references}

@inproceedings{AbboudBackurs2017,
  author    = {Amir Abboud and Arturs Backurs},
  title     = {Towards Hardness of Approximation for Polynomial Time Problems},
  booktitle = {8th Innovations in Theoretical Computer Science Conference (ITCS 2017)},
  series    = {Leibniz International Proceedings in Informatics},
  volume    = {67},
  pages     = {11:1--11:26},
  publisher = {Schloss Dagstuhl--Leibniz-Zentrum f{\"u}r Informatik},
  year      = {2017},
  doi       = {10.4230/LIPIcs.ITCS.2017.11},
  url       = {https://drops.dagstuhl.de/entities/document/10.4230/LIPIcs.ITCS.2017.11}
}

@inproceedings{AbboudRubinstein2018,
  author    = {Amir Abboud and Aviad Rubinstein},
  title     = {Fast and Deterministic Constant Factor Approximation Algorithms for {LCS} Imply New Circuit Lower Bounds},
  booktitle = {9th Innovations in Theoretical Computer Science Conference (ITCS 2018)},
  series    = {Leibniz International Proceedings in Informatics},
  volume    = {94},
  pages     = {35:1--35:14},
  publisher = {Schloss Dagstuhl--Leibniz-Zentrum f{\"u}r Informatik},
  year      = {2018},
  doi       = {10.4230/LIPIcs.ITCS.2018.35},
  url       = {https://drops.dagstuhl.de/entities/document/10.4230/LIPIcs.ITCS.2018.35}
}

@inproceedings{BenSassonViola2014,
  author    = {Eli Ben-Sasson and Emanuele Viola},
  title     = {Short {PCPs} with Projection Queries},
  booktitle = {41st International Colloquium on Automata, Languages, and Programming (ICALP 2014)},
  series    = {Lecture Notes in Computer Science},
  volume    = {8572},
  pages     = {163--173},
  publisher = {Springer},
  year      = {2014},
  doi       = {10.1007/978-3-662-43948-7_14},
  url       = {https://doi.org/10.1007/978-3-662-43948-7_14}
}

@inproceedings{ChenEtAl2019,
  author    = {Lijie Chen and Shafi Goldwasser and Kaifeng Lyu and
               Guy N. Rothblum and Aviad Rubinstein},
  title     = {Fine-grained Complexity Meets {IP = PSPACE}},
  booktitle = {Proceedings of the Thirtieth Annual ACM-SIAM Symposium on Discrete Algorithms (SODA)},
  pages     = {1--20},
  publisher = {SIAM},
  year      = {2019},
  doi       = {10.1137/1.9781611975482.1},
  url       = {https://doi.org/10.1137/1.9781611975482.1}
}

@inproceedings{BringmannKunnemann2015,
  author    = {Karl Bringmann and Marvin K{\"u}nnemann},
  title     = {Quadratic Conditional Lower Bounds for String Problems and Dynamic Time Warping},
  booktitle = {56th Annual IEEE Symposium on Foundations of Computer Science (FOCS)},
  pages     = {79--97},
  publisher = {IEEE},
  year      = {2015},
  doi       = {10.1109/FOCS.2015.15},
  url       = {https://doi.org/10.1109/FOCS.2015.15}
}

@inproceedings{Kuszmaul2019,
  author    = {William Kuszmaul},
  title     = {Dynamic Time Warping in Strongly Subquadratic Time: Algorithms for the Low-Distance Regime and Approximate Evaluation},
  booktitle = {46th International Colloquium on Automata, Languages, and Programming (ICALP 2019)},
  series    = {Leibniz International Proceedings in Informatics},
  volume    = {132},
  pages     = {80:1--80:15},
  publisher = {Schloss Dagstuhl--Leibniz-Zentrum f{\"u}r Informatik},
  year      = {2019},
  doi       = {10.4230/LIPIcs.ICALP.2019.80},
  url       = {https://drops.dagstuhl.de/entities/document/10.4230/LIPIcs.ICALP.2019.80}
}

@inproceedings{AbboudBackursWilliams2015,
  author    = {Amir Abboud and Arturs Backurs and Virginia Vassilevska Williams},
  title     = {Tight Hardness Results for {LCS} and Other Sequence Similarity Measures},
  booktitle = {56th Annual IEEE Symposium on Foundations of Computer Science (FOCS)},
  pages     = {59--78},
  publisher = {IEEE},
  year      = {2015},
  doi       = {10.1109/FOCS.2015.14},
  url       = {https://doi.org/10.1109/FOCS.2015.14}
}

@inproceedings{BackursIndyk2016,
  author    = {Arturs Backurs and Piotr Indyk},
  title     = {Which Regular Expression Patterns Are Hard to Match?},
  booktitle = {57th IEEE Symposium on Foundations of Computer Science (FOCS)},
  pages     = {457--466},
  publisher = {IEEE},
  year      = {2016},
  doi       = {10.1109/FOCS.2016.56},
  url       = {https://doi.org/10.1109/FOCS.2016.56}
}

@inproceedings{BennettPeikertTang2022,
  author    = {Huck Bennett and Chris Peikert and Yi Tang},
  title     = {Improved Hardness of {BDD} and {SVP} Under {Gap-(S)ETH}},
  booktitle = {13th Innovations in Theoretical Computer Science Conference (ITCS 2022)},
  series    = {Leibniz International Proceedings in Informatics},
  volume    = {215},
  pages     = {19:1--19:12},
  publisher = {Schloss Dagstuhl--Leibniz-Zentrum f{\"u}r Informatik},
  year      = {2022},
  doi       = {10.4230/LIPIcs.ITCS.2022.19},
  url       = {https://drops.dagstuhl.de/entities/document/10.4230/LIPIcs.ITCS.2022.19}
}

@inproceedings{XiKuszmaul2022,
  author    = {Zoe Xi and William Kuszmaul},
  title     = {Approximating Dynamic Time Warping Distance Between Run-Length Encoded Strings},
  booktitle = {30th Annual European Symposium on Algorithms (ESA 2022)},
  series    = {Leibniz International Proceedings in Informatics},
  volume    = {244},
  pages     = {90:1--90:19},
  publisher = {Schloss Dagstuhl--Leibniz-Zentrum f{\"u}r Informatik},
  year      = {2022},
  doi       = {10.4230/LIPIcs.ESA.2022.90},
  url       = {https://drops.dagstuhl.de/entities/document/10.4230/LIPIcs.ESA.2022.90}
}

@article{FroeseEtAl2023,
  author    = {Vincent Froese and Brijnesh J. Jain and Maciej Rymar and Mathias Weller},
  title     = {Fast Exact Dynamic Time Warping on Run-Length Encoded Time Series},
  journal   = {Algorithmica},
  volume    = {85},
  number    = {2},
  pages     = {492--508},
  year      = {2023},
  doi       = {10.1007/s00453-022-01038-3},
  url       = {https://link.springer.com/article/10.1007/s00453-022-01038-3}
}

@inproceedings{AgarwalEtAl2016,
  author    = {Pankaj K. Agarwal and Kyle Fox and Jiangwei Pan and Rex Ying},
  title     = {Approximating Dynamic Time Warping and Edit Distance for a Pair of Point Sequences},
  booktitle = {32nd International Symposium on Computational Geometry (SoCG 2016)},
  series    = {Leibniz International Proceedings in Informatics},
  volume    = {51},
  pages     = {6:1--6:16},
  publisher = {Schloss Dagstuhl--Leibniz-Zentrum f{\"u}r Informatik},
  year      = {2016},
  doi       = {10.4230/LIPIcs.SoCG.2016.6},
  url       = {https://drops.dagstuhl.de/entities/document/10.4230/LIPIcs.SoCG.2016.6}
}

@inproceedings{BonehEtAl2024,
  author    = {Itai Boneh and Shay Golan and Shay Mozes and Oren Weimann},
  title     = {{$\widetilde{O}$}ptimal Dynamic Time Warping on Run-Length Encoded Strings},
  booktitle = {51st International Colloquium on Automata, Languages, and Programming (ICALP 2024)},
  series    = {Leibniz International Proceedings in Informatics},
  volume    = {297},
  pages     = {30:1--30:17},
  publisher = {Schloss Dagstuhl--Leibniz-Zentrum f{\"u}r Informatik},
  year      = {2024},
  doi       = {10.4230/LIPIcs.ICALP.2024.30},
  url       = {https://drops.dagstuhl.de/entities/document/10.4230/LIPIcs.ICALP.2024.30}
}

@inproceedings{BonehGolanKociumaka2026,
  author    = {Itai Boneh and Shay Golan and Tomasz Kociumaka},
  title     = {Dynamic Time Warping in the Low-Distance Regime},
  booktitle = {67th Annual IEEE Symposium on Foundations of Computer Science (FOCS)},
  publisher = {IEEE},
  year      = {2026},
  url       = {https://arxiv.org/abs/2609.40249},
  note      = {To appear}
}

@misc{Kuszmaul2021Binary,
  author       = {William Kuszmaul},
  title        = {Binary Dynamic Time Warping in Linear Time},
  year         = {2021},
  howpublished = {arXiv:2101.01108v2},
  url          = {https://arxiv.org/abs/2101.01108}
}

@article{ImpagliazzoPaturi2001,
  author    = {Russell Impagliazzo and Ramamohan Paturi},
  title     = {On the Complexity of {$k$-SAT}},
  journal   = {Journal of Computer and System Sciences},
  volume    = {62},
  number    = {2},
  pages     = {367--375},
  year      = {2001},
  doi       = {10.1006/jcss.2000.1727},
  url       = {https://doi.org/10.1006/jcss.2000.1727}
}

@article{Williams2005,
  author    = {Ryan Williams},
  title     = {A New Algorithm for Optimal 2-Constraint Satisfaction and Its Implications},
  journal   = {Theoretical Computer Science},
  volume    = {348},
  number    = {2--3},
  pages     = {357--365},
  year      = {2005},
  doi       = {10.1016/j.tcs.2005.09.023},
  url       = {https://doi.org/10.1016/j.tcs.2005.09.023}
}

@article{AlonEtAl1992,
  author    = {Noga Alon and Oded Goldreich and Johan H{\aa}stad and Ren{\'e} Peralta},
  title     = {Simple Constructions of Almost {$k$}-wise Independent Random Variables},
  journal   = {Random Structures \& Algorithms},
  volume    = {3},
  number    = {3},
  pages     = {289--304},
  year      = {1992},
  doi       = {10.1002/rsa.3240030308},
  url       = {https://doi.org/10.1002/rsa.3240030308}
}

@article{BravermanEtAl2020,
  author    = {Vladimir Braverman and Moses Charikar and William Kuszmaul and
               Lin F. Yang},
  title     = {The One-Way Communication Complexity of Dynamic Time Warping Distance},
  journal   = {Journal of Computational Geometry},
  volume    = {11},
  number    = {2},
  pages     = {62--93},
  year      = {2020},
  doi       = {10.20382/jocg.v11i2a4},
  url       = {https://jocg.org/index.php/jocg/article/view/3119}
}

@inproceedings{Bringmann2014Frechet,
  author    = {Karl Bringmann},
  title     = {Why Walking the Dog Takes Time: {F}r{\'e}chet Distance Has No Strongly Subquadratic Algorithms Unless {SETH} Fails},
  booktitle = {55th Annual IEEE Symposium on Foundations of Computer Science (FOCS)},
  pages     = {661--670},
  publisher = {IEEE},
  year      = {2014},
  doi       = {10.1109/FOCS.2014.76},
  url       = {https://doi.org/10.1109/FOCS.2014.76}
}

@article{BringmannMulzer2016,
  author    = {Karl Bringmann and Wolfgang Mulzer},
  title     = {Approximability of the Discrete {F}r{\'e}chet Distance},
  journal   = {Journal of Computational Geometry},
  volume    = {7},
  number    = {2},
  pages     = {46--76},
  year      = {2016},
  doi       = {10.20382/jocg.v7i2a4},
  url       = {https://jocg.org/index.php/jocg/article/view/3000}
}

@article{ChakrabortyEtAl2020,
  author    = {Diptarka Chakraborty and Debarati Das and Elazar Goldenberg and Michal Kouck{\'y} and Michael Saks},
  title     = {Approximating Edit Distance Within Constant Factor in Truly Sub-quadratic Time},
  journal   = {Journal of the ACM},
  volume    = {67},
  number    = {6},
  pages     = {36:1--36:22},
  year      = {2020},
  doi       = {10.1145/3422823},
  url       = {https://doi.org/10.1145/3422823}
}

@inproceedings{GoldenbergRubinsteinSaha2020,
  author    = {Elazar Goldenberg and Aviad Rubinstein and Barna Saha},
  title     = {Does Preprocessing Help in Fast Sequence Comparisons?},
  booktitle = {Proceedings of the 52nd Annual ACM SIGACT Symposium on Theory of Computing (STOC)},
  pages     = {657--670},
  publisher = {ACM},
  year      = {2020},
  doi       = {10.1145/3357713.3384300},
  url       = {https://doi.org/10.1145/3357713.3384300}
}

@article{AndoniNosatzki2026,
  author    = {Alexandr Andoni and Negev Shekel Nosatzki},
  title     = {Edit Distance in Near-Linear Time: It's a Constant Factor},
  journal   = {SIAM Journal on Computing},
  volume    = {55},
  number    = {4},
  pages     = {708--786},
  year      = {2026},
  doi       = {10.1137/21M1392322},
  url       = {https://doi.org/10.1137/21M1392322}
}

@inproceedings{BuchinOpheldersSpeckmann2019,
  author    = {Kevin Buchin and Tim Ophelders and Bettina Speckmann},
  title     = {{SETH} Says: Weak {F}r{\'e}chet Distance Is Faster, but Only If It Is Continuous and in One Dimension},
  booktitle = {Proceedings of the Thirtieth Annual ACM-SIAM Symposium on Discrete Algorithms (SODA)},
  pages     = {2887--2901},
  publisher = {SIAM},
  year      = {2019},
  doi       = {10.1137/1.9781611975482.179},
  url       = {https://doi.org/10.1137/1.9781611975482.179}
}

@inproceedings{VanDerHorstOphelders2026,
  author    = {van der Horst, Thijs and Ophelders, Tim},
  title     = {A Strongly-Subquadratic {$(3+\varepsilon)$}-Approximation for the {F}r{\'e}chet Distance for Paths in Metric Spaces},
  booktitle = {34th Annual European Symposium on Algorithms (ESA 2026)},
  series    = {Leibniz International Proceedings in Informatics},
  volume    = {388},
  pages     = {111:1--111:17},
  publisher = {Schloss Dagstuhl--Leibniz-Zentrum f{\"u}r Informatik},
  year      = {2026},
  doi       = {10.4230/LIPIcs.ESA.2026.111},
  url       = {https://drops.dagstuhl.de/entities/document/10.4230/LIPIcs.ESA.2026.111}
}

@inproceedings{MaoRubinstein2026,
  author    = {Xiao Mao and Aviad Rubinstein},
  title     = {Approximation Schemes for Edit Distance and {LCS} in Quasi-Strongly Subquadratic Time},
  booktitle = {Proceedings of the 58th Annual ACM Symposium on Theory of Computing (STOC)},
  pages     = {744--754},
  publisher = {ACM},
  year      = {2026},
  doi       = {10.1145/3798129.3800789},
  url       = {https://doi.org/10.1145/3798129.3800789}
}

@misc{MaderTavasoliWang2026,
  author       = {Ethan Mader and Borna Tavasoli and Jihan Wang},
  title        = {A Strongly Subquadratic {$(3+\varepsilon)$}-Approximation for Weighted Edit Distance over Arbitrary Metrics},
  year         = {2026},
  howpublished = {arXiv:2609.14873},
  url          = {https://arxiv.org/abs/2609.14873}
}

@misc{DasKipouridisKociumaka2026,
  author       = {Debarati Das and Evangelos Kipouridis and Tomasz Kociumaka},
  title        = {Metric Weighted Edit Distance: {$(3+\varepsilon)$}-Approximation in {$\widetilde O_\varepsilon(N^{1.6})$} Time},
  year         = {2026},
  howpublished = {arXiv:2609.20796},
  url          = {https://arxiv.org/abs/2609.20796}
}

@misc{Gorbachev2026,
  author       = {Egor Gorbachev},
  title        = {Factor Three Approximation for Edit Distance},
  year         = {2026},
  howpublished = {arXiv:2610.01311},
  url          = {https://arxiv.org/abs/2610.01311}
}

@article{ImpagliazzoPaturiZane2001,
  author    = {Russell Impagliazzo and Ramamohan Paturi and Francis Zane},
  title     = {Which Problems Have Strongly Exponential Complexity?},
  journal   = {Journal of Computer and System Sciences},
  volume    = {63},
  number    = {4},
  pages     = {512--530},
  year      = {2001},
  doi       = {10.1006/jcss.2001.1774},
  url       = {https://doi.org/10.1006/jcss.2001.1774}
}
\endgroup
\end{document}